\documentclass[11pt]{article}
\pdfoutput=1 
\usepackage[sc]{mathpazo}

\usepackage[margin=1in]{geometry}
\usepackage[english]{babel}
\usepackage[utf8]{inputenc}
\usepackage[compact]{titlesec}

\usepackage{cmap}
\usepackage[T1]{fontenc}
\usepackage{bm}
\usepackage{amsfonts}
\usepackage{amsmath}
\usepackage{amssymb}
\usepackage{amsthm}
\usepackage{float}
\usepackage{graphics}
\usepackage[normalem]{ulem}

\usepackage{hyperref}
\usepackage[svgnames, table]{xcolor}
\hypersetup{colorlinks={true},urlcolor={blue},linkcolor={DarkBlue},citecolor=[named]{DarkGreen},linktoc=all}
\usepackage{natbib}

\usepackage{microtype}
\usepackage[capitalise,nameinlink,noabbrev]{cleveref}

\usepackage{doi}

\usepackage{booktabs} 
\usepackage{amsfonts,amsbsy,amsthm}
\usepackage{thmtools,thm-restate}
\usepackage{enumerate}
\usepackage{bbm}
\usepackage{nicefrac}
\usepackage{wrapfig}
\usepackage[justification = centering]{subcaption}
\usepackage{mathtools}
\usepackage{array,multirow}
  {\list{}{\leftmargin=#1\rightmargin=#1}\item[]}%
  {\endlist}

\newtheorem{theorem}{Theorem}
\newtheorem{lemma}[theorem]{Lemma}

\newtheorem{proposition}[theorem]{Proposition}
\newtheorem{corollary}[theorem]{Corollary}

\newtheorem{claim}[theorem]{Claim}
\newtheorem{example}[theorem]{Example}

\newcommand{\EF}[1]{\ifstrempty{#1}{\textrm{\textup{EF}}}{\textrm{\textup{EF{$#1$}}}}}

\newcommand{\EQ}[1]{\ifstrempty{#1}{\textrm{\textup{EQ}}}{\textrm{\textup{EQ{$#1$}}}}}
\newcommand{\EQX}{\textrm{\textup{EQX}}}
\newcommand{\PoF}{\textup{PoF}}
\newcommand{\PoE}{\textup{PoE}}

\renewcommand{\ge}{\geqslant}
\renewcommand{\le}{\leqslant}

\colorlet{mygray}{gray!40}
\colorlet{mycolor}{blue!15}

\usepackage{tikz}
\usetikzlibrary{calc,arrows,matrix,decorations.pathreplacing}
\usepackage{paralist}
\usepackage{algorithm}
\usepackage[noend]{algpseudocode}
\usepackage{algorithmicx}

\usepackage[capitalise,noabbrev]{cleveref}
\Crefname{claim}{Claim}{Claims}
\Crefname{corollary}{Corollary}{Corollaries}
\Crefname{definition}{Definition}{Definitions}
\Crefname{example}{Example}{Examples}
\Crefname{lemma}{Lemma}{Lemmas}
\Crefname{property}{Property}{Properties}
\Crefname{proposition}{Proposition}{Propositions}
\Crefname{remark}{Remark}{Remarks}
\Crefname{theorem}{Theorem}{Theorems}

\title{The Price of Equity with Binary Valuations\\ and Few Agent Types}
\author{
	\begin{tabular}{m{0.12\textwidth}m{0.12\textwidth}m{0.12\textwidth}m{0.12\textwidth}m{0.12\textwidth}m{0.12\textwidth}}
		\multicolumn{3}{c}{\textbf{Umang Bhaskar}} & \multicolumn{3}{c}{\textbf{Neeldhara Misra}}\\
		\multicolumn{3}{c}{\small{Tata Institute of Fundamental Research}} & \multicolumn{3}{c}{\small{Indian Institute of Technology Gandhinagar}}\\
		\multicolumn{3}{c}{\small{Mumbai, Maharashtra, India}} & \multicolumn{3}{c}{\small{Palaj, Gujarat, India}}\\		\multicolumn{3}{c}{\href{mailto:umang@tifr.res.in}{\small{\texttt{umang@tifr.res.in}}}} & \multicolumn{3}{c}{\href{mailto:neeldhara.m@iitgn.ac.in}{\small{\texttt{neeldhara.m@iitgn.ac.in}}}}\\
		&&&&&\\
		\multicolumn{3}{c}{\textbf{Aditi Sethia}} & \multicolumn{3}{c}{\textbf{Rohit Vaish}}\\
		\multicolumn{3}{c}{\small{Indian Institute of Technology Gandhinagar}} & \multicolumn{3}{c}{\small{Indian Institute of Technology Delhi}}\\
            \multicolumn{3}{c}{\small{Palaj, Gujarat, India}} & \multicolumn{3}{c}{\small{New Delhi, Delhi, India}}\\
		\multicolumn{3}{c}{\href{mailto:aditi.sethia@iitgn.ac.in}{\small{\texttt{aditi.sethia@iitgn.ac.in}}}} & \multicolumn{3}{c}{\href{mailto:rvaish@iitd.ac.in}{\small{\texttt{rvaish@iitd.ac.in}}}}\\
	\end{tabular}
}

\date{}

\begin{document}

\maketitle 

\begin{abstract}
In fair division problems, the notion of price of fairness measures the loss in welfare due to a fairness constraint. Prior work on the price of fairness has focused primarily on envy-freeness up to one good (EF1) as the fairness constraint, and on the utilitarian and egalitarian welfare measures. Our work instead focuses on the price of equitability up to one good (EQ1) (which we term \emph{price of equity}) and considers the broad class of \emph{generalized $p$-mean} welfare measures (which includes utilitarian, egalitarian, and Nash welfare as special cases). We derive fine-grained bounds on the price of equity in terms of the \emph{number of agent types} (i.e., the maximum number of agents with distinct valuations), which allows us to identify scenarios where the existing bounds in terms of the number of agents are overly pessimistic.

Our work focuses on the setting with binary additive valuations, and obtains upper and lower bounds on the price of equity for $p$-mean welfare for all $p \leqslant 1$. For any fixed $p$, our bounds are tight up to constant factors. A useful insight of our work is to identify the \emph{structure} of allocations that underlie the upper (respectively, the lower) bounds \emph{simultaneously} for all $p$-mean welfare measures, thus providing a unified structural understanding of price of fairness in this setting. This structural understanding, in fact, extends to the more general class of binary submodular (or matroid rank) valuations. We also show that, unlike binary additive valuations, for binary submodular valuations the number of agent types does not provide bounds on the price of equity.
\end{abstract}

\section{Introduction}
\label{sec:Intro}

Tradeoffs are inevitable when we pursue multiple optimization objectives that are typically not simultaneously achievable. Quantifying such tradeoffs is a fundamental problem in computation, game theory, and economics. Our focus in this work is on the ``price of fairness'' in the context of fair division problems, which is a notion that captures tradeoffs between \emph{fairness} and \emph{welfare}. 

Recall that a fair division instance in the discrete setting involves a set of $n$ \emph{agents} $N = \{1,2,\dots,n\}$, $m$ indivisible \emph{goods} $M = \{g_1, \ldots, g_m\}$, and $\mathcal{V} \coloneqq \{v_1,v_2,\dots,v_n\}$, a \emph{valuation profile} consisting of each agent's valuation of the goods. For any agent $i \in N$, its valuation function $v_i: 2^M \rightarrow \mathbb{N} \cup \{0\}$ specifies its numerical value (or \emph{utility}) for every subset of goods in $M$. We will assume that the valuations are normalised, that is, for all $i \in N$, $v_i(M) = W$, where $W$ is the normalisation constant. Our goal is to devise an \emph{allocation} of goods to agents; defined as an ordered partition\footnote{Unless otherwise specified, we implicitly assume that allocations are \emph{complete}, i.e., every good is assigned to some agent.} of the $m$ goods into $n$ ``bundles'', where the bundles are (possibly empty) subsets of $M$, and the convention is that the $i^{th}$ bundle in the partition is the set of goods assigned to the agent $i$. 

The \emph{welfare} of an allocation is a measure of the utility that the agents derive from the allocation. For additive valuations, the individual utility that an agent $i$ derives from their bundle $A_i$ is simply the sum of the values that they have for the goods in $A_i$. The overall welfare of an allocation $A$ is typically defined by aggregating individual utilities in various ways. Not surprisingly, there are several notions of welfare corresponding to different approaches to consolidating the individual utilities. For instance, the \emph{utilitarian social welfare} is the sum of utilities of agents under $A$; the \emph{egalitarian social welfare} is the lowest utility achieved by any agent with respect to $A$; and the \emph{Nash social welfare} is the geometric mean of utilities of agents under $A$. One may view all of these welfare notions as special cases of the $p$-mean welfare (where $p \in (-\infty,0) \cup (0,1]$), which is defined as the generalized $p$-mean of utilities of agents under $A$, i.e., $W_{p}(A) \coloneqq \left( \frac{1}{n} \sum_{i \in N} \bigl(v_i(A_i)\bigr)^p \right)^{\nicefrac{1}{p}}.$ Note that for $p>1$, the $p$-mean optimal allocation tends to concentrate the distribution among fewer agents (consider the simple case of two identical agents with additive valuations who value each of two goods at 1), which is contrary to the spirit of fairness. Hence we focus on $p \leqslant 1$.


A natural goal for a fair division problem is to obtain an allocation that maximizes the overall welfare. However, observe that optimizing exclusively for welfare can lead to undesirable allocations. To see this, consider an instance with additive valuations where all the valuation functions are the same, i.e., the utility of any good $g$ is the same for all agents in $N$. In this case, \emph{every} allocation has the same utilitarian welfare. 
So, when we only optimize for---in this example, utilitarian---welfare, we have no way of distinguishing between, say, the allocation that allocates all goods to one agent and one that distributes the goods more evenly among the agents. To remedy this, one is typically interested in allocations that not only maximize welfare, but are also ``fair''. 

There are several notions of fairness studied in the literature. Consider an allocation $A = (A_1,\dots,A_n)$. We say that $A$ is \emph{envy-free} (\EF{}) if for any pair of agents $i$ and $k$, we have that $i$ values $A_i$ at least as much as they value $A_k$, i.e., $v_i(A_i) \geqslant v_i(A_k)$; and \emph{equitable} (\EQ{}) if every pair of agents $i$ and $k$ value their respective bundles equally, i.e., $v_i(A_i) = v_k(A_k)$. While both these fairness goals are natural, they may not be achievable, such as in a trivial instance with one good valued positively by two agents. This has motivated several approximations, and in particular, the notions of \emph{envy-freeness up to one good} and \emph{equitability up to one good} have been widely studied. The allocation $A$ is \emph{envy-free up to one good} (\EF{1}) if for any pair of agents $i,k \in N$ such that $A_k \neq \emptyset$, there is a good $g \in A_k$ such that $v_i(A_i) \geqslant v_i(A_k \setminus \{g\})$. Analogously, $A$ is \emph{equitable up to one good} (\EQ{1}) if for any pair of agents $i,k \in N$ such that $A_k \neq \emptyset$, there is a good $g \in A_k$ such that $v_i(A_i) \geqslant v_k(A_k \setminus \{g\})$. For instances with additive valuations (and somewhat beyond), \EF{1} (and, similarly, \EQ{1}) allocations are guaranteed to exist.

The price of fairness is informally the cost of achieving a specific fairness notion, where the cost is viewed through the lens of a particular welfare concept. For a fairness notion $\mathcal{F}$ (such as \EQ{1} or \EF{1}) and a welfare notion $\mathcal{W}$ (such as egalitarian or utilitarian welfare), the price of fairness is the ``worst-case ratio'' 
of the maximum welfare (measured by $\mathcal{W}$) that can be obtained by \emph{any} allocation, to the maximum welfare that can be obtained among allocations that are fair according to $\mathcal{F}$. For example, it is known from the work of~\cite{CKM+19unreasonable} that under additive valuations, any allocation that maximizes the Nash social welfare satisfies EF1. Thus, the price of fairness of \EF{1} with respect to Nash social welfare is $1$. Further,~\cite{BBS20optimal} show that the price of EF1 with respect to utilitarian welfare is $O(\sqrt{n})$ for normalised subadditive valuations.

In this contribution, we focus on bounds for the price of fairness in the context of \EQ{1}, a notion that we will henceforth refer to as the \emph{price of equity} (PoE) when there is no ambiguity. Much of the existing literature on price of fairness analysis focuses on \emph{specific} welfare measures (e.g., utilitarian, egalitarian, and Nash social welfare). Our work deviates from this trend by analyzing the \emph{entire} family of generalized $p$-mean welfare measures (i.e., for \emph{all} $p \leqslant 1$); recall that this captures the notions of egalitarian, utilitarian, and Nash welfare as special cases. Our results therefore address the behavior of the price of equity for a wide spectrum of welfare notions.

Further, we obtain bounds in terms of the \emph{number of agent types} --- which we denote by $r$ --- rather than the total number of agents. The number of agent types of a fair division instance is the largest number of agents whose valuations are mutually distinct: in other words, it is the number of distinct valuation functions in the instance. Note that the number of agent types is potentially \emph{much} smaller than the total number of agents. The notion of agent types has been popular in the fair division literature for the reason that it is a natural quantification of the ``simplicity'' of the structure of the instance as given by the valuations. Note that the well-studied special case of identical valuations is equivalent to the class of instances for which $r = 1$, and therefore one might interpret parameterizing by $r$ as a smooth generalization of the case of identical valuations. For a representative selection of studies that focus on instances with a bounded number of agent types, we refer the reader to~\citep{BliemBN16,BouveretCEIP17,GargKM21,BranzeiLM16}.

We restrict ourselves to the setting of \emph{binary submodular} (also known as matroid rank) valuations. A valuation function $v_i$ is submodular if for any subsets of goods $S, S' \subseteq M$ such that $S \subseteq S'$, and for any good $g \not \in S'$, $v_i(S \cup g) - v_i(S) \geqslant v_i(S' \cup g) - v_i(S')$. That is, the marginal value of adding $g$ to $S$ is at least that of adding $g$ to a superset of $S$. Valuation $v_i$ is binary submodular if for any subset of goods $S \subseteq M$ and any good $g$, the marginal value $v_i(S \cup g) - v_i(S) \in \{0,1\}$. Binary submodular valuations are frequently studied in fair division and are considered to be a useful special case such as in allocating items under a budget, or with exogenous quotas~\citep{BCI+21finding,VZ23general}. It also provides algorithmic leverage: many computational questions of interest that are hard in general turn out to be tractable once we restrict our attention to binary submodular valuations. As an example, while it is NP-hard to compute a Nash social welfare maximizing allocation even for identical additive valuations~\citep{RoosR10}, such an allocation can be computed in polynomial time under binary submodular valuations in conjunction with other desirable properties such as strategyproofness, envy-freeness up to any good, and ex-ante envy-freeness~\citep{BEF21fair}.

A strict subset of binary submodular valuations is the class of \emph{binary additive} valuations---this is a subclass of additive valuations wherein each value $v_i(g)$ is either $0$ or $1$. 
Binary additive valuations provide a simple way for agents to express their preferences as they naturally align with the idea of agents ``approving'' or ``rejecting'' goods. These are also widely studied in the literature on fair division, for example, see~\citep{O20multi,KSV20almost,BEF21fair,ABF+21maximum,AR21almost}. In the case of voting too, binary additive valuations play a role.~\cite{DS15maximizing} consider the complexity of maximizing Nash social welfare when scores inherent in classical voting procedures are used to associate utilities with the agents' preferences, and find that the case of approval ballots --- which happen to lead to binary additive valuations --- are a tractable subclass. 

\subsection*{Our Contributions and Techniques}

\renewcommand{\arraystretch}{2}
\begin{table*}[t]
\centering
\begin{tabular}{|rr|cc|cc|}
 \hline
 \multicolumn{2}{|c|}{\multirow{2}{*}{\textbf{\PoE{}}}} & \multicolumn{2}{c|}{\textbf{Agent types ($r$)}} \\
 \cline{3-4}
 & & Lower bound & Upper bound\\
\hline
  %
  %
  \multicolumn{2}{|c|}{Utilitarian welfare ($p=1$)} & \cellcolor{mycolor}$r-1$ & \cellcolor{mycolor}$r$ \\
  %
  \hline
  %
  \multicolumn{2}{|c|}{Nash welfare ($p = 0$)} & \cellcolor{mycolor}$\frac{(r-1)/e}{\ln (r-1)}$ & \cellcolor{mycolor}$\frac{(r-1)}{\ln (r-1)/e}$ \\
  %
  \hline
  %
  \multicolumn{2}{|c|}{Egalitarian welfare ($p\rightarrow-\infty$)} & 1 & 1~\citep{SCD23equitability} \\
  \hline
  %
  \multicolumn{2}{|c|}{$p \in (0,1)$} & \cellcolor{mycolor}$p(r-1)/e$ & \cellcolor{mycolor}$2r - 1$ \\
  %
  \hline
  \multicolumn{2}{|c|}{$p \in (-1,0)$} & \cellcolor{mycolor} $2^{1/p} (r-1)^{1/(1-p)}$ & \cellcolor{mycolor}  $2^{-1/p} (-p)^{1/p(1-p)} (r-1)^{1/(1-p)} $  \\
  %
  \hline
   \multicolumn{2}{|c|}{$p \leqslant -1$} & \cellcolor{mycolor} $2^{1/p} (r-1)^{1/(1-p)}$ & \cellcolor{mycolor}  $2 (r-1)^{1/(1-p)} $  \\
  \hline
\end{tabular}
\vspace{0.1in}
\caption{Summary of results for the price of equity (\PoE{}). Each cell indicates either the lower or the upper bound  (columns) on \PoE{} for a specific welfare measure (rows) as a function of the number of \emph{agent types} $r$. Our contributions are highlighted by shaded boxes. The lower bounds are from~\Cref{thm:lbdgoodstwo}, while the upper bounds are shown in~\Cref{thm:ubrank} and~\Cref{thm:pof-upperbounds}. ~\Cref{sec:figures} in the appendix presents the upper and lower bounds graphically as a function of $r$, for $p=1$, $p=0$, $p=-1$, and $p=-10$.}
\label{tab:Results}
\end{table*}

We now turn to a discussion of our findings (see~\Cref{tab:Results} for a summary of our results for binary additive valuations). Given an instance of fair division with binary submodular valuations, let $A^\star$ be an allocation that maximizes the Nash social welfare. It is implicit from the results of~\cite{BCI+21finding} that $A^\star$ also has maximum $p$-mean welfare for all $p \leqslant 1$ (for details, refer to~\Cref{subsec:Opt-pmean-allocations}). We show an analogous result for \EQ{1} allocations, by demonstrating that there exists an \EQ{1} allocation (which we call $B$, or the \emph{truncated allocation}) that maximizes the $p$-mean welfare for all $p$. To this end, in allocation $A^\star$, let $i$ be an agent with minimum value, and let $\ell = v_i(A_i^\star)$. If the allocation is not already \EQ{1}, then we reallocate ``excess'' goods from the bundles of agents who value their bundles at more than $\ell + 1$ and give them to agent $i$. Notice that agent $i$ must have marginal value $0$ for all these excess goods, otherwise this reallocation would improve the Nash welfare. It turns out that this allocation $B$ is \EQ{1} and also has --- among \EQ{1} allocations --- the highest $p$-mean welfare. 

\begin{restatable}[]{theorem}{bisoptimal}
    \label{thm:b-is-optimal}
    For any $p \in \mathbb{R} \cup \{- \infty\}$  and binary submodular valuations, the $p$-mean welfare of the truncated allocation $B$ is at least that of any other \EQ{1} allocation.
\end{restatable}

Notice that together with the result of~\cite{BCI+21finding},~\Cref{thm:b-is-optimal} allows us to focus only on the maximum Nash social welfare allocation $A^\star$ and the truncated allocation $B$ to obtain upper bounds on the PoE for all $p \leqslant 1$ simultaneously. 

We now describe our bounds on the PoE for binary additive valuations. Our lower bounds are based on varying parameters in a single basic instance. The parameters are $r$, the number of agent types, and $W$, the normalisation constant for the agents. Given $r$ and $W$, the instance has $m = rW$ goods, divided into $r$ groups of $W$ goods each. The groups are $M_1$, $M_2$, $\ldots$, $M_r$. There are $W+1$ agents who value all the goods in $M_1$ at $1$ each and everything else at $0$. Further, for each $2 \leqslant i \leqslant r$, we have exactly one agent who values the goods in $M_i$ and nothing else.

To summarize, we have $W+1$ agents of the first type, who have a common interest in $W$ goods. Any allocation must leave one of these agents with zero value. Beyond these coveted goods, each of the remaining goods is valued by exactly one agent. A welfare maximizing allocation will 
allocate each good in $M_2 \cup \cdots \cup M_r$ to the unique agent who values it; however, an \EQ{1} allocation is constrained by the fact that an agent of the first type must get value $0$.\footnote{For $p \leqslant 0$, we use the standard convention that allocation $A$ is a $p$-mean optimal allocation if (a) it maximizes number of agents with positive value, and (b) among all allocations that satisfy (a), maximizes the $p$-mean welfare when restricted to agents with positive value.} It turns out that using this family of instances, we can obtain the following bounds. 

\begin{restatable}[\textbf{PoE lower bounds}]{theorem}{lbdgoodstwo}
    Let $s := r - 1$. The price of equity for binary additive valuations is at least:
    \begin{enumerate}
        \item $s$, for $p=1$,
        \item $\frac{p}{e}s$, for $p \in (0,1)$,
        \item $\frac{s}{e \ln s}$, for $p = 0$,
        \item $2^{1/p} s^{1/(1-p)}$, for $p < 0$.
    \end{enumerate}
    \label{thm:lbdgoodstwo}
\end{restatable}

We now turn to the upper bounds for binary additive valuations. It turns out that the PoE for utilitarian welfare is bounded by the \emph{rank} of the instance, where the rank is simply the rank of the $n \times m$ matrix $\{v_{i}(g_j)\}_{1\leqslant i \leqslant n; 1\leqslant j \leqslant m}$. Observe that the rank is a lower bound for the number of agent types, so this result also implies an upper bound of $r$ on the PoE. In fact, the rank could be logarithmic in the number of agent types, and hence this is a significantly tighter bound than the number of agent types.

To obtain this upper bound, in allocation $B$ (which, as shown in \Cref{thm:b-is-optimal}, maximizes the utilitarian welfare among \EQ{1} allocations) we 
show that the number of wasted goods is at most $m(1-\frac{1}{k})$, where $k$ is the rank of the instance. This implies the theorem.

\begin{restatable}[\textbf{Utilitarian PoE upper bound}]{theorem}{ubrank}
    Under binary additive valuations and utilitarian welfare as the objective, the price of equity is at most the rank of the instance.
    \label{thm:ubrank}
\end{restatable}

For other values of $p$, we obtain the following upper bounds.

\begin{restatable}[\textbf{PoE upper bounds}]{theorem}{pofupperbounds}
    Let $s := r - 1$. The price of equity for binary additive valuations is at most
    \begin{enumerate}
        \item $1+s$ for $p=1$
        \item $1 + 2s$ for $p \in (0,1)$
        \item $\frac{s}{\ln (s/e)}$ for $p=0$ (i.e., the Nash social welfare)
        \item $s^{1/(1-p)} 2^{-1/p} (-1/p)^{1/p(p-1)}$ for $p \in (-1,0)$
        \item $2 s^{1/(1-p)}$ for $p \leqslant -1$
    \end{enumerate}
    \label{thm:pof-upperbounds}
\end{restatable}

We note that for any fixed $p$, the lower bounds (\Cref{thm:lbdgoodstwo}) and upper bounds (\Cref{thm:pof-upperbounds}) are 
within a constant factor of each other.

Conceptually, for the proof of the upper bounds, we show that the worst case for the PoE is in fact the family of instances used for showing our lower bounds in~\Cref{thm:lbdgoodstwo}. In particular, any instance can be transformed into one belonging to the lower bound family, without improving the PoE. Note that for the PoE, we can focus on the allocations $A^\star$ and $B$ irrespective of the $p$-mean welfare measure, since these maximize the $p$-mean welfare for all $p \leqslant 1$ simultaneously. For a given instance, let $l$ be the minimum value of any agent in $A^\star$. We divide the agent types into two groups: types for which every agent has value at most $l+1$ in $A^\star$, and types for which an agent has value $> l+1$. Note that for a type in the first group, each agent of this type retains its value in $B$, while for a type in the second group, the value of each agent of this type is truncated to $l+1$. Our proof shows that agents in the first group must have total value at least $W$, as in the lower bound example. We also use $W$ as an upper bound for the total value of each agent type in the second group. Then letting $\alpha$ be the fraction of agents in the first group, and optimizing over $\alpha$, gives us the required upper bounds.

We then consider the PoE for binary additive valuations with the additional structure that both the rows and the columns are normalised. That is, each agent values exactly $W$ goods, and each good is valued by exactly $W_c$ agents. For such \emph{doubly normalised} instances, we show the PoE is 1.

\begin{restatable}[]{theorem}{doublynormalised}
For doubly normalised instances under binary additive valuations, the PoE for the $p$-mean welfare is 1 for all $p\le 1$.
\label{thm:doubly-normalised}
\end{restatable}

Finally, we obtain bounds on the PoE for binary submodular valuations. For identical valuations, it follows from similar results for \EF{1} that the PoE is 1.

\begin{restatable}[]{proposition}{mrvpofonetype}
    When all agents have identical binary submodular valuations, the PoE is 1 for p-mean welfare measure for all $p \le 1$.
    \label{prop:mrv-pof-1}
\end{restatable}

However, this is the best that can be obtained, in the sense that even with just \emph{two} agent types, the PoE for utilitarian welfare is at least $n/6$, where $n$ is the number of agents. Hence we cannot obtain bounds on the PoE that depend on the number of agent types, as we did for binary additive valuations.

\begin{restatable}[]{theorem}{pofmrvtwotypes}
    The PoE for utilitarian welfare when agents have binary submodular valuations is at least $n/6$ (where $n$ is the number of agents), even when there are just two types of agents.
    \label{thm:pofmrvtwotypes}
\end{restatable}

Nonetheless, we do obtain an upper bound of $2n$ on the PoE for binary submodular valuations.

\begin{restatable}[]{theorem}{pofmrvubn}
    For binary submodular valuations and any $p \leqslant 1$, the PoE for $p$-mean welfare is at most $2n$.
    \label{thm:pofmrvubn}
\end{restatable}

\subsection*{Related Work}

The notion of \emph{price of fairness} was proposed in the works of \cite{BFT11price} and~\cite{CKK+12efficiency}. These formulations were inspired from similar notions in game theory---specifically, \emph{price of stability} and \emph{price of anarchy}---that capture the loss in social welfare due to strategic behavior.\footnote{Price of anarchy was defined by~\cite{KP09worst} and subsequently studied in the notable work of \cite{RT02bad}, while price of stability was defined by~\cite{ADK+08price}.} \cite{CKK+12efficiency} studied the price of fairness for divisible and indivisible resources under three fairness notions: \emph{proportionality}~\citep{S48division}, \emph{envy-freeness}~\citep{GS58puzzle,F67resource}, and \emph{equitability}~\citep{DS61cut}. For indivisible resources, they defined price of fairness only with respect to those instances that admit some allocation satisfying the fairness criterion.

Recently, \cite{BLM+21price} studied price of fairness for indivisible goods for fairness notions whose existence is guaranteed; in particular, they studied \emph{envy-freeness up to one good} (\EF{1}) and \emph{maximum Nash welfare} allocations.\footnote{The \EF{1} notion was formulated by~\cite{B11combinatorial} although subsequently it was observed that an algorithm of~\cite{LMM+04approximately} achieves this guarantee for monotone valuations. The Nash social welfare function was originally proposed in the context of the bargaining problem~\citep{N50bargaining} and subsequently studied for resource allocation problems by~\cite{EG59consensus}.} In a similar vein,~\cite{SCD21connections} studied price of fairness for allocating indivisible \emph{chores} for different relaxations of envy-freeness and maximin share. 
Perhaps closest to our work is a recent paper by~\cite{SCD23equitability}. This work studies price of equity and price of equitability for any item (\EQX{}) for indivisible goods as well as indivisible chores under utilitarian and egalitarian welfare. The valuations are assumed to be additive but not necessarily binary. For indivisible goods, the price of equity is shown to be between $n-1$ and $3n$, where $n$ is the number of agents, while for egalitarian welfare, a tight bound of $1$ is provided.

\section{Preliminaries}
\label{sec:Preliminaries}

\paragraph{\textbf{Problem instance.}} An instance of the fair division problem is specified by a tuple $\langle N, M, \mathcal{V} \rangle$, where $N = \{1,2,.\dots,n\}$ is a set of $n \in \mathbb{N}$ \emph{agents}, $M = \{g_1, \ldots, g_m\}$ is a set of $m$ indivisible \emph{goods}, and $\mathcal{V} \coloneqq \{v_1,v_2,\dots,v_n\}$ is the \emph{valuation profile} consisting of each agent's valuation function. For any agent $i \in N$, its valuation function $v_i: 2^M \rightarrow \mathbb{N} \cup \{0\}$ specifies its numerical value (or \emph{utility}) for every subset of goods in $M$. For simplicity, for a valuation function $v$, we will denote $v(\{g\})$ as $v(g)$.

Agents $i$ and $j$ are said to be of the same \emph{type} if their valuation functions are identical, i.e., if for every subset of goods $S \subseteq M$, $v_i(s) = v_j(S)$. We will use $r$ to denote the number of distinct agent types in an instance. Further, an instance is \emph{normalised} if for some constant $W$, $v_i(M) = W$ for all agents $i$. Our work focuses on instances with normalised valuations, since there are trivial instances where the price of equity for any $p$-mean welfare for $p \in \mathbb{R}$ is large without this assumption (e.g., the simple instance with 2 agents and $k$ goods, where agent 1 has value 1 for the first good and zero for the others, and agent 2 has value $1$ for all goods, has price of equity $k/3$ for the utilitarian welfare). 

\paragraph{\textbf{Classes of valuation functions.}} A valuation function $v$ is:
\begin{itemize}
    \item \emph{monotone} if for any two subsets of goods $S$ and $T$ such that $S \subseteq T$, we have $v(S) \leqslant v(T)$,
    \item \emph{monotone submodular} (or simply submodular) if it is monotone and for any two subsets of goods $S$ and $T$ such that $S \subseteq T$ and any good $g \in M \setminus T$, we have $v(S \cup \{g\}) - v(S) \geqslant v(T \cup \{g\}) - v(T)$,
    \item \emph{additive} if for any subset of goods $S \subseteq M$, we have $v(S) = \sum_{g \in S} v(g)$,
    \item \emph{binary submodular} (or \emph{matroid rank}) if it is submodular and for any subset $S \subseteq M$ and any good $g \in M \setminus S$, we have $v(S \cup \{g\}) - v(S) \in \{0,1\}$, and
    \item \emph{binary additive} if it is additive and for any good $g \in M$, $v(\{g\}) \in \{0,1\}$.
\end{itemize}
The containment relations between these classes are as follows:
$$\textup{Binary additive} \subseteq \textup{Additive} \subseteq \textup{Submodular} \subseteq \textup{Monotone}$$
and
$$\textup{Binary additive} \subseteq \textup{Binary submodular} \subseteq \textup{Submodular} \subseteq \textup{Monotone}.$$
The domains of additive and binary submodular valuations are incomparable in the sense that an instance belonging to one class may not belong to the other.

We will primarily focus on binary submodular valuations in \Cref{sec:optimal-p-mean} and~\ref{sec:binary-submodular}, and on binary additive valuations in \Cref{sec:bav-lb},~\ref{sec:bav-ub}, and~\ref{sec:doubly-normalised}. 

\paragraph{\textbf{Allocation.}} A \emph{bundle} refers to any (possibly empty) subset of goods. An allocation $A \coloneqq (A_1,\dots,A_n)$ is a partition of the set of goods $M$ into $n$ bundles; here, $A_i$ denotes the bundle assigned to agent $i$.

Given an allocation $A$, we say that agent $i$ values good $g$ if $v_i(A_i \cup \{g\}) > v_i(A_i \setminus \{g\})$. Thus if $g \in A_i$, then removing $g$ decreases $i$'s value. Else, assigning $g$ to $A_i$ increases $i$'s value. For additive valuations, specifying an allocation is unnecessary, and we say $i$ values $g$ if $v_i(\{g\}) = 1$.

Further, for an allocation $A$, we say a good $g \in A_i$ is \emph{wasted} if $v_i(A_i \setminus \{g\}) = v_i(A_i)$, i.e., if removing it does not change the value of agent $i$. For additive valuations, this implies that $v_i(g) = 0$. We say an allocation (possibly partial) is \emph{wasteful} if some good is wasted (and is \emph{non-wasteful} or \emph{clean} otherwise). If $A$ is a clean allocation, then for binary submodular valuations, for each agent $i$, $v_i(A_i) = |A_i|$.

\paragraph{\textbf{Fairness notions.}} An allocation $A = (A_1,\dots,A_n)$ is said to be:
\begin{itemize}
    \item \emph{envy-free} (\EF{}) if for any pair of agents $i,k \in N$, we have $v_i(A_i) \geqslant v_i(A_k)$~\citep{GS58puzzle,F67resource}, 
    \item \emph{envy-free up to one good} (\EF{1}) if for any pair of agents $i,k \in N$ such that $A_k \neq \emptyset$, there is a good $g \in A_k$ such that $v_i(A_i) \geqslant v_i(A_k \setminus \{g\})$~\citep{B11combinatorial,LMM+04approximately},
    \item \emph{equitable} (\EQ{}) if for any pair of agents $i,k \in N$, we have $v_i(A_i) = v_k(A_k)$~\citep{DS61cut}, and
    \item \emph{equitable up to one good} (\EQ{1}) if for any pair of agents $i,k \in N$ such that $A_k \neq \emptyset$, there is a good $g \in A_k$ such that $v_i(A_i) \geqslant v_k(A_k \setminus \{g\})$~\citep{GMT14near,FSV+19equitable}.
\end{itemize}

Note that if all agents are identical (i.e., they are of the same type), then the notions of \EF{1} and \EQ{1} (and similarly, \EF{} and \EQ{}) coincide.

\paragraph{\textbf{Pareto optimality.}} An allocation $A = (A_1,\dots,A_n)$ is said to be Pareto dominated by another allocation $B = (B_1,\dots,B_n)$ if for every agent $i \in N$, $v_i(B_i) \geqslant v_i(A_i)$ and for some agent $k \in N$, $v_k(B_k) > v_k(A_k)$. A \emph{Pareto optimal} allocation is one that is not Pareto dominated by any other allocation.


\paragraph{\textbf{Welfare measures.}} We will now discuss various welfare measures associated with an allocation $A$.
\begin{itemize}
    \item \emph{Utilitarian social welfare} is the sum of utilities of agents under $A$, i.e., $\mathcal{W}^\texttt{util}(A) \coloneqq \sum_{i \in N} v_i(A_i)$.
    \item \emph{Egalitarian social welfare} is the utility of the least-happy agent under $A$, i.e., $\mathcal{W}^\texttt{egal}(A) \coloneqq \min_{i \in N} v_i(A_i)$.
    \item \emph{Nash social welfare} is the geometric mean of utilities of agents under $A$, i.e., $\mathcal{W}^\texttt{Nash}(A) \coloneqq \left( \Pi_{i \in N} v_i(A_i) \right)^{\nicefrac{1}{n}}$, and
    \item for any $p \in \mathbb{R}$, the \emph{$p$-mean welfare} is the generalized $p$-mean of utilities of agents under $A$, i.e., $\mathcal{W}_{p}(A) \coloneqq \left( \frac{1}{n} \sum_{i \in N} v_i^p(A_i) \right)^{\nicefrac{1}{p}}$.
\end{itemize}

For any $p \in \mathbb{R}$, the $p$-mean welfare is a strictly increasing, symmetric function of the agent values. It can be observed that utilitarian, egalitarian, and Nash welfare are all special cases of $p$-mean welfare for $p=1$, $p \rightarrow -\infty$, and $p \rightarrow 0$, respectively~\citep{M04fair}.

Given an instance, it may be that in every allocation, some agent gets zero value. In this case, we need to redefine the $p$-mean welfare. We fix a largest subset of agents $S$ that can simultaneously get positive value in an allocation, and then define $\mathcal{W}_p(A) = \left( \frac{1}{|S|} \sum_{i \in S} v_i^p(A_i) \right)^{\nicefrac{1}{p}}$. This follows prior work on the Nash welfare, e.g.,~\citep{CKM+19unreasonable,BCI+21finding}.

A \emph{leximin} allocation is one which maximizes the minimum utility, then subject to that, it maximizes the second minimum, and so on. Thus a leximin allocation also maximizes the egalitarian welfare. 

\paragraph{\textbf{Price of fairness.}} Given a fairness notion $\mathcal{F}$ (e.g., \EF{1}, \EQ{1}) and a $p$-mean welfare measure, the price of fairness of $\mathcal{F}$ with respect to a welfare measure $\mathcal{W}_p$ is the supremum over all fair division instances with $n$ agents and $m$ goods of the ratio of the maximum welfare (according to $\mathcal{W}_p$) of any allocation and the maximum welfare of any allocation that satisfies $\mathcal{F}$. 

Formally, let $\mathcal{I}_{n,m}$ denote the set of all fair division instances with $n$ agents and $m$ items. Let $\mathcal{A}(I)$ denote the set of all allocations in the instance $I$, and further let $\mathcal{A}_{\mathcal{F}}(I)$ denote the set of all allocations in the instance $I$ that satisfy the fairness notion $\mathcal{F}$.


Then, the price of fairness (\PoF{}) of the fairness notion $\mathcal{F}$ with respect to the welfare measure $\mathcal{W}_p$ is defined as:
$$\PoF(\mathcal{F},\mathcal{W}_p) \coloneqq \sup_{I \in \mathcal{I}_{n,m}} \frac{ \max_{A^* \in \mathcal{A}(I)} \mathcal{W}_p(A^*)}{  \max_{B \in \mathcal{A}_{\mathcal{F}}(I)} \mathcal{W}_p(B)}.$$

As indicated earlier, throughout this paper we will focus on equitability up to one good (\EQ{1}) as the fairness notion of choice  (i.e., $\mathcal{F}$ is $\EQ{1}$). For notational simplicity, we will just write $\PoF$ instead of $\PoF(\mathcal{F},\mathcal{W})$ whenever the welfare measure $\mathcal{W}$ is clear from context, and we will refer to this ratio as the price of equity (PoE) whenever the fairness notion in question is \EQ{1}.


\subsection{Some properties of $p$-mean welfare}

We state here some basic properties of the $p$-mean welfare that will be useful in due course. 

\begin{restatable}[]{claim}{pmeanconcave}
    For all $p < 1$, the $p$-mean welfare is a concave function of the agent valuations.
    \label{claim:pmean-is-concave}
\end{restatable}

This proof was shown by~\cite{Ahle20Concave}. We reproduce it here for completeness.

\begin{proof}
    
    Consider $f(x)=\left(\Sigma_i x_i^p\right)^{1 / p}$. The Hessian matrix $H$ is then given by:
    $$
    H_{i j} =(1-p) f^{1-2 p} A \quad \text { where } A_{i j}= \begin{cases}-x_i^{p-2} \Sigma_{k \neq i} x_k^p & \text { if } i=j \\ x_i^{p-1} x_j^{p-1} & \text { if } i \neq j\end{cases}
    $$
    For $p \le 1$, the matrix $H$ is negative semidefinite, since the initial coefficient $(1-p) f^{1-2p} \ge 0$, and for any vector $v$,
    $$
    v^T A v ~=~ \left(\sum_{i=1}^n v_i x_i^{p-1}\right)^2-\sum_{i=1}^n v_i^2 x_i^{p-2} \sum_j x_i^p ~ \le ~ 0
    $$
    where the last inequality follows by applying the Cauchy-Schwarz inequality\footnote{For any vectors $a$, $b$, $(\sum_i a_i b_i)^2 \le (\sum_i a_i^2) \times (\sum_i b_i^2)$} to $\left(v_i x^{p / 2-1}\right) \cdot\left(x_i^{p / 2}\right)$. Hence, the function $f$ is concave.
\end{proof}

\begin{corollary}
Given a vector of values for $n$ agents $x \in \mathbb{R}_+^n$ and a subset $S \subseteq N$ of agents, let $x'$ be the vector where $x_i' = x_i$ if $i \not \in S$, and $x_i' = \sum_{j \in S} x_j/|S|$ if $i \in S$. Then for all $p \leqslant 1$, 

\[
\left(\frac{1}{n} \sum_{i=1}^n (x_i)^p\right)^{1/p} \leqslant \left(\frac{1}{n} \sum_{i=1}^n (x_i')^p\right)^{1/p} \, ,
\]

\noindent i.e., averaging out the value for a subset of agents weakly increases the $p$-mean welfare.
\label{corollary:averaging-subset}
\end{corollary}

\begin{restatable}[]{claim}{usefulinequality}
    Given $l \in \mathbb{N}$, and a vector $(x_1, \ldots, x_l) \in \mathbb{R}_+^l$, for $p \in [0,1]$,

    \[
    \frac{1}{l} \sum_{i=1}^l x_i^{1-p} \leqslant \left(\frac{1}{l} \sum_{i=1}^l x_i\right)^{1-p} \, ,
    \]

    \noindent while for $p < 0$, the opposite inequality holds.
    \label{claim:averaging-values}
\end{restatable}

\begin{proof}
    For $p \in \{0,1\}$, the claim can be seen by simply substituting these values. For $p \in (0,1)$, the function $f(x) = x^{1-p}$ is concave, hence an application of Jensen's inequality gives us the claim. For $p < 0 $, the function $f(x) = x^{1-p}$ is convex, hence again, Jensen's inequality gives us the claim. 
\end{proof}

\section{Optimal allocations for binary submodular valuations}
\label{sec:optimal-p-mean}

We first show that for obtaining bounds on the price of equity for the class of binary submodular valuations (and hence, for binary additive valuations), we can focus on two allocations: the first is the Nash welfare optimal allocation $A^\star$, which obtains the optimal $p$-mean welfare for all $p \leqslant 1$, and the second is the truncated allocation $B$, which obtains the optimal $p$-mean welfare among all \EQ{1} allocations for all $p \in \mathbb{R} \cup \{-\infty\}$.

\subsection{An optimal $p$-mean welfare allocation}
\label{subsec:Opt-pmean-allocations}

\cite{BCI+21finding} show the following results.

\begin{proposition}[\citealp{BCI+21finding}, Theorem 3.14]
Let $\Phi:\mathbb{Z}^n \rightarrow \mathbb{R}$ be a symmetric strictly convex function, and let $\Psi:\mathbb{Z}^n \rightarrow \mathbb{R}$ be a symmetric strictly concave function. Let $A$ be some allocation. For binary submodular valuations, the following statements are equivalent:
\begin{enumerate}
    \item $A$ is a minimizer of $\Phi$ over all the utilitarian optimal allocations,
    \item $A$ is a maximizer of $\Psi$ over all the utilitarian optimal allocations,
    \item $A$ is a leximin allocation, and
    \item $A$ maximizes Nash social welfare.
\end{enumerate}
\label{prop:nsw-is-optimal}
\end{proposition}

\begin{proposition}[\citealp{BCI+21finding}, Theorem 3.11]
For binary submodular valuations, any Pareto optimal allocation is utilitarian optimal.
\label{prop:po-means-utilitarian-optimal}
\end{proposition}

For $p \le 1$, if the $p$-mean welfare function was strictly concave, then it would follow immediately that the Nash welfare optimal allocation $A^\star$ in fact simultaneously maximizes the $p$-mean welfare for all $p \le 1$. However, in general the $p$-mean welfare is concave (\Cref{claim:pmean-is-concave}), but not strictly concave. E.g., for any $p \le 1$ and any vector of values $v=(v_1, \ldots, v_n)$ with $v_i > 0$ for all agents $i$, let us overload notation slightly and define $\mathcal{W}_p(v) = \left(\frac{1}{n} \sum_{i=1}^n v_i^p \right)^{1/p}$. Then $\mathcal{W}_p(2v) = (\mathcal{W}_p(v) + \mathcal{W}_p(3v))/2$, violating strict concavity. However, we can slightly modify the proof of Theorem 3.14 from~\cite{BCI+21finding}, to obtain the following result. The modified proof is in \Cref{subsec:NashOpt} in the Appendix.

\begin{restatable}[]{proposition}{NashOptSimultaneouslyMaximizespmean}
For binary submodular valuations, any Nash welfare maximizing allocation (and hence, leximin allocation) simultaneously maximizes the $p$-mean welfare for all $p \leqslant 1$.
\label{prop:NashOpt_Simultaneously_Maximizes_p-mean}
\end{restatable}




\subsection{An optimal $p$-mean welfare \EQ{1} allocation}
\label{subsec:Truncated_Allocation}

We now show that similarly, there exists an \EQ{1} allocation $B$ that maximizes the $p$-mean welfare for all $p$. Given $A^\star$, allocation $B$ is obtained as follows, which we call the \emph{truncated allocation}. Let $l = \min_i v_i(A_i^\star)$ be the smallest value that any agent obtains in $A^\star$, and let $i_l$ be an agent that has this minimum value. Note that for any agent $i$, if $v_i(A_i^\star) \geqslant l+2$, then all goods allocated to $i$ must have marginal value $0$ for the agent $i_l$, i.e., for all $g \in A_i^\star$, $v_{i_l}(A_{i_l}^\star \cup \{g\}) = v_{i_l}(A_{i_l}^\star)$  (else we can increase the Nash social welfare by re-allocating any good that violates this to agent $i_l$). 

For the \EQ{1} allocation that we would like to construct, for any agent $i$ with $v_i(A_i^\star) \geqslant l+2$, we remove goods from $A_i^\star$ until $i$'s value for the remaining bundle is $l+1$. We allocate the removed goods to agent $i_l$ (that has marginal value $0$ for these goods). Let $B$ be the resulting allocation. Then clearly, if $v_i(A_i^\star) \in \{l,l+1\}$, then $v_i(B_i) = v_i(A_i^\star)$, else $v_i(B_i) = l+1$. Thus, allocation $B$, our truncated NSW allocation, is \EQ{1}.



\bisoptimal*
\begin{proof}
    Let $n_1$ be the number of agents that have value $l$ in allocation $B$, and $n_2$ be the number of agents with value $l+1$. Clearly, $n = n_1 + n_2$. Consider any other allocation $C$. We will show that the following statement is true: either (i) there exists an agent $i$ with $v_i(C_i) \leqslant l-1$, or (ii) if all agents have value $v_i(C_i) \geqslant l$, then at most $n_2$ agents have value $\geqslant l+1$ (and hence at least $n_1$ agents have value $ \leqslant l$). 
    
    Assuming the statement is true, if $C$ is an \EQ{1} allocation, either (i) every agent has value $\leqslant l$, or (ii) at most $n_2$ agents have value $l+1$, and at least $n_1$ agents have value $\leqslant l$. It follows that allocation $B$ maximizes any symmetric non-decreasing function of agent valuations in the set of \EQ{1} allocations, and hence $B$ maximizes the $p$-mean welfare among all \EQ{1} allocations for all $p \in \mathbb{R}$. Since the minimum agent valuation in $B$ is the same as in $A^\star$, which by~\Cref{prop:nsw-is-optimal} also maximizes the egalitarian welfare, allocation $B$ maximizes the $p$-mean welfare for $p = -\infty$ as well.
    
    Lastly, to prove the statement, by the truncation procedure that yields allocation $B$, the number of agents $|\{i: v_i(A_i^\star) \geqslant l+1\}|$ that have value at least $l+1$ in allocation $A^\star$ is also $n_2$. Further, by \Cref{prop:nsw-is-optimal}, $A^\star$ is also a leximin allocation, and hence no allocation in which every agent has value at least $l$, can have more than $n_2$ agents with value at least $l+1$. The statement follows.
\end{proof}

\section{Lower Bounds on the PoE for Binary Additive Valuations}
\label{sec:bav-lb}


\lbdgoodstwo*


Note that as $p \rightarrow -\infty$, $2^{1/p} s^{1/(1-p)} \rightarrow 1$. 

\begin{proof}
All our lower bounds are based on varying parameters in a single instance. The parameters are $r$, the number of agent types, and $W$, the normalisation constant for the agents. Given $r$, $W$, the instance has $m = rW$ goods, divided into $r$ groups of $W$ goods each. The groups are $M_1$, $M_2$, $\ldots$, $M_r$. There are $W+1$ agents of the first agent type, and 1 agent each of the remaining $r-1$ types (thus, $n = W+r$). Agents of type $t$ have value $1$ for the goods in group $M_t$, and value $0$ for all other goods. The instance is thus \emph{disjoint}; no good has positive value for agents of two different types.


    We note that following properties of our lower bound instance:

    \begin{enumerate}
        \item For any $p \leqslant 1$, an optimal $p$-mean welfare allocation has value $1$ for $W$ agents of the first type, and value $W$ for each of the remaining $r-1$ agents. 
        \item For any $p \leqslant 1$, the \EQ{1} allocation with maximum $p$-mean welfare gives value $1$ to all agents except for one agent of the first type (since there are $W+1$ agents of the first type, and only $W$ goods for which they have positive value).
    \end{enumerate}

    We use $\Lambda_p$ to denote the PoE for this instance. Then $\Lambda_p$ is exactly 

    \begin{align*}
    \Lambda_p ~ = ~    \displaystyle \left(\frac{\frac{1}{W+r-1} \left(W \times 1^p + (r-1) \times W^p\right)}{\frac{1}{W+r-1}\left(W \times 1^p + (r-1) \times 1^p\right)}\right)^{1/p} ~ = ~ \displaystyle \left(\frac{W + s \times W^p}{W+s}\right)^{1/p} \, .
    \end{align*}

    Note that although there are $W+r$ agents, in any allocation one agent must have value $0$, hence the $p$-mean average is taken over $W+r-1$ agents. For each of the cases in the theorem, we will now show how to choose $W$, $s$ to obtain the bound claimed.

    For $p=1$, choose $W = s^2$. Then 

    \begin{align*}
        \Lambda_p ~ \geqslant ~ \frac{W + sW}{ W + s} ~ = ~ \frac{s^2 + s^3}{s^2 + s} ~ = ~ s \, ,
    \end{align*}

    \noindent giving the required bound.

    For $p \in (0,1)$, choose $W = p s$. Then

    \begin{align*}
        \displaystyle \Lambda_p & \geqslant ~  \left(\frac{W + s \times W^p}{W+s}\right)^{1/p} \\
         & = \left(\frac{ps + s \times (ps)^p}{ps+s}\right)^{1/p} ~ = ~ \left(\frac{p + (ps)^p}{p+1}\right)^{1/p} \\
         & \geqslant ps (p+1)^{-1/p} \geqslant ps / e \, , \qquad \text{since $1+x \leqslant e^x$} \, .
    \end{align*}

    For $p = 0$, the $p$-mean welfare is the Nash social welfare. Note that in the \EQ{1} allocation, each of $W+s$ agents has value $1$, hence the NSW is $1$. In the optimal Nash social welfare allocation, $W$ agents have value $1$, and $s$ agents have value $W$, hence the NSW is $W^{s/W+s}$, which is also the PoE for this instance. Now choose $W = s / \ln s$. Then 
    
   \begin{align*}
        \displaystyle \Lambda_p & \geqslant ~  \exp{\frac{s \ln W}{s+W}} ~ = ~  \exp{\frac{s ( \ln s - \ln \ln s)}{s+s/\ln s}} ~ = ~ \exp{\frac{\ln s - \ln \ln s}{1+1/\ln s}} \\
            & \geqslant ~ \exp{\frac{\ln s - \ln \ln s}{1+1/(\ln s - \ln \ln s - 1)}} \\
            & = ~ \exp{(\ln s - \ln \ln s -1)} ~ = ~ \frac{s}{e \ln s} \, .
    \end{align*}

    Lastly, for $p < 0$, choose $W$ so that $W = s W^p$, or $W = s^{1/(1-p)}$. Then

    \begin{align*}
    \Lambda_p & = ~ \displaystyle \left(\frac{W + s \times W^p}{W+s}\right)^{1/p} \\
        & = ~ \displaystyle \left(\frac{2W}{W+s}\right)^{1/p} \\
        & = ~ \displaystyle 2^{1/p}\left(\frac{W^p}{1 + W^p}\right)^{1/p} \\
        & \geqslant ~ \displaystyle 2^{1/p} s^{1/(1-p)}  \, ,
    \end{align*}

    \noindent where the last inequality is because $p < 0$.
\end{proof}

\section{Upper Bounds on the PoE for Binary Additive Valuations}
\label{sec:bav-ub}

We first consider the case of utilitarian welfare, and then present our results for $p < 1$.

\subsection{Upper bounds on the PoE for $p=1$}

We assume that each good has value $1$ for at least one agent, else the good can be removed without consequence. Given an instance with binary additive valuations for the agents, for an agent $i$, we overload notation and let $v_i \coloneqq (v_i(g))_{g \in M}$ denote the vector of values for the individual goods. Define $V$ to be the matrix whose $i$th row is given by $v_i$.


We say that an instance has rank $k$ if the matrix $V$ has rank $k$ (equivalently, there are $k$ linearly independent valuation vectors among the agents). Note that the rank is a lower bound on both the number of agent types, as well as the number of good types. Finally, since the rank is $k$, we assume the agents are ordered so that the vectors $v_1$, $\ldots$, $v_k$ are linearly independent; the corresponding agents are called basis agents.

\ubrank*

\begin{proof}
Let $k$ denote the rank of the instance, and consider allocation $B$ that maximizes the utilitarian welfare among all \EQ{1} allocations. Recall that a good $g$ is \emph{wasted} if it is assigned to agent $i$ such that $v_i(g) = 0$. We will show that the number of wasted goods is at most $m(1-\frac{1}{k})$. Thus, allocation $A$ has social welfare at least $m/k$. Since the optimal social welfare is at most $m$, this would be sufficient to prove the theorem.

Since allocation $A$ is \EQ{1}, there exists a utility level $\ell$ such that for each agent $i$, $v_i(A_i) \in \{\ell, \ell+1\}$. We say an agent $i$ is \emph{poor} if $v_i(A_i) = \ell$, else agent $i$ is \emph{rich}. If $v_i(A_i) = \ell$ for all agents, then all agents are poor.

Suppose for a contradiction that strictly more than $m(1-\frac{1}{k})$ goods are wasted. Consider a wasted good $g$ and a poor agent $i$. It must be true that $v_i(g) = 0$, else we could assign $g$ to $i$ and increase the utilitarian welfare while maintaining \EQ{1}. Hence if agent $i$ is poor, then $v_i(g) = 0$ for each wasted good $g$. Hence, $v_i(g) = 1$ for strictly less than $m/k$ goods. Then due to normalisation, every agent has value $1$ for strictly less than $m/k$ goods. In particular, the $k$ basis agents have value $1$ for strictly less than $m/k$ goods each. Thus, there is a good --- say $g^*$ --- for which each basis agent has value $0$.

By definition, the value of each agent for $g^*$ is a linear combination of the values of the basis agents for $g^*$. Since the basis agents have value $0$ for $g^*$, it follows that every agent must have value $0$ for $g^*$, yielding the required contradiction. 
\end{proof}

It follows immediately from the theorem that the price of equity is also bounded by the number of agent types.

\begin{restatable}{corollary}{}
Under binary additive valuations and utilitarian welfare as the objective, the price of equity is at most $r$, the number of agent types.
\label{cor:Util-UpperBound-Agenttypes}
\end{restatable}

\subsection{Upper bounds on the PoE for $p < 1$}
\label{subsec:Upper-Bounds-PoE}

From~\Cref{prop:nsw-is-optimal} and~\Cref{thm:b-is-optimal}, to bound the PoE for any $p<1$, it suffices to obtain an upper bound on the ratio of the $p$-mean welfare for the two allocations $A^\star$ (which maximizes the Nash welfare) and $B$ (the truncated allocation). 

We will use various properties of the allocations $A^\star$ and $B$ in the following proofs. To state these, define $T_k$ as the set of agents of type $k$, and let $S_k$ be the set of goods allocated to agents in $T_k$ by $A^\star$. That is, $S_k := \cup_{i \in T_k} A_i^\star$. Let $m_k := |S_k|$, and $n_k := |T_k|$. Then note that for each agent $i \in T_k$, 

\[
\displaystyle v_i(A_i^\star) = |A_i^\star| \in \left\{\Bigl\lfloor\dfrac{m_k}{n_k} \Bigr\rfloor,\left\lceil\dfrac{m_k}{n_k}\right\rceil \right\} \,.
\]

We reindex the types in increasing order of the averaged number of goods assigned by $A^\star$, so that $m_i/n_i \leqslant m_{i+1}/n_{i+1}$. Now define
\[
\lambda := \displaystyle 
\begin{cases}
\lceil \frac{m_1}{n_1} \rceil & \text{if $m_1/n_1$ is fractional} \\
1 + \frac{m_1}{n_1} & \text{if $m_1/n_1$ is integral} \, .
\end{cases}
\]

Thus $\lambda$ is integral, $\lambda > m_1/n_1$, and $\lambda \ge 2$ (since the $p$-mean welfare is only taken over agents with positive valuation, $m_1 \ge n_1$). Note that in $A^\star$, the smallest value of any agent is $\lfloor m_1/n_1 \rfloor$, and $\lambda \leqslant 1 + \lfloor m_1/n_1 \rfloor$. Hence agents with value at most $\lambda$ in $A^\star$ will retain their value in allocation $B$, by definition of $B$, while other agents will have their values truncated to $\lambda$.

Now let $\rho$ be the highest index so that $\lambda \geqslant m_\rho/n_\rho$. Thus, 
\begin{align}
    \lambda \geqslant \frac{\sum_{i=1}^\rho m_i}{\sum_{i=1}^\rho n_i} \, .
    \label{align:lambda-sum}
\end{align}

As stated above, any agent of type $k \leqslant \rho$ will retain their value, i.e., $v_i(B_i) = v_i(A_i^\star)$ for an agent $i$ of type $k \leqslant \rho$. 

We claim that agents of the first $\rho$ types must have at least $W$ goods assigned to them in $A^\star$.

\begin{claim}
    $\sum_{i=1}^\rho m_i \geqslant W$.
    \label{claim:items-in-first-rho-types}
\end{claim}

\begin{proof}
    For a contradiction, let $\sum_{i=1}^\rho m_i < W$. Since $\lambda > m_1/n_1$, there is an agent $i^\star$ of type 1 with value $v_i(A_i^\star) = \lambda - 1$. Since $\sum_{i=1}^\rho m_i < W$, a good $g$ that has value $1$ for agents of type 1 is allocated in $A^\star$ to an agent $i'$ of type $k > \rho$. Since $m_k / n_k > \lambda$ by definition of $\rho$, there is an agent $i''$ of type $k$ with value $v_{i''}(A_{i''}^\star) \geqslant \lambda + 1$. Since $A^\star$ maximizes the Nash social welfare, any good $h \in A_{i''}^\star$ has value 1 for both agents $i''$ and $i'$. Then it is easy to see that transferring any good from $i''$ to $i'$, and then transferring good $g$ from $i'$ to $i^\star$, will increase the Nash social welfare. Since $A^\star$ maximizes the Nash social welfare, we have a contradiction.
\end{proof}

Then from~\eqref{align:lambda-sum} and~\Cref{claim:items-in-first-rho-types}, we obtain
\begin{align}
    \lambda \geqslant W / \sum_{i=1}^\rho n_i \, . \label{align:lambda-sum-2}
\end{align}

We now obtain a general expression for bounding the PoE for all $p \leqslant 1$. We will then optimize this expression for different ranges of $p$, to obtain upper bounds on the PoE.

\begin{restatable}[]{lemma}{upperboundexpression}

    The price of equity for $p$-mean welfare for instances with $r$ types is at most
    \begin{enumerate}
        \item $\sup_{\alpha \in [0,1]} \left(\alpha + \alpha^p s^p (1-\alpha)^{1-p}\right)^{1/p}$ for $p < 0$
        \item $\sup_{\alpha \in [0,1]} \left(\frac{s \alpha}{1-\alpha}\right)^{(1-\alpha)}$ for $p = 0$,
        \item $\sup_{\alpha \in [0,1]} \left(\alpha + 2^p\alpha^p s^p (1-\alpha)^{1-p}\right)^{1/p}$ for $p \in (0,1)$.
    \end{enumerate}

    \noindent where as before, $s = r-1$.
    \label{lem:pof-ub-base}
\end{restatable}

\begin{proof}
    The $p$-mean welfare for the NSW optimal allocation $A^\star$ is
    \begin{align*}
        \mathcal{W}_p(A^\star) & = ~ \displaystyle \left( \frac{1}{n} \sum_{i=1}^n v_i(A_i^\star)^p \right)^{1/p}
        ~ = ~  \left( \frac{1}{n} \sum_{k=1}^r \sum_{i \in T_k} v_i(A_i^\star)^p \right)^{1/p} \, ,
    \end{align*}
    
    \noindent where in the last expression, we partition the agents by their respective types. 
    
    We now consider the agent types $k \leqslant \rho$ and $k > \rho$ separately. For agents of type $k > \rho$, we average out the values and replace their individual values by the average value $m_k/n_k$, and use~\Cref{corollary:averaging-subset} to obtain
    \begin{align*}
        \mathcal{W}_p(A^\star) & \leqslant ~  \displaystyle \left( \frac{1}{n} \left( \sum_{k=1}^\rho \sum_{i \in T_k} v_i(A_i^\star)^p + \sum_{k=\rho+1}^r n_k \left(\frac{m_k}{n_k}\right)^p \right) \right)^{1/p} \, .
    \end{align*}
    
    The truncated allocation $B$ is an \EQ{1} allocation, and we will consider the ratio $\mathcal{W}_p(A^\star)/\mathcal{W}_p(B)$. This is clearly an upper bound on the price of equity. For allocation $B$, recall that for agents $i$ of type $k \leqslant \rho$, $v_i(B_i) = v_i(A_i^\star)$ since these are not truncated, while for agents $i$ of type $k > \rho$, $v_i(B_i) = \lambda$. Hence the PoE is 
    \begin{align}
        \frac{\mathcal{W}_p(A^\star)}{\mathcal{W}_p(B)} & \leqslant ~  \displaystyle \left( \frac{\sum_{k=1}^\rho \sum_{i \in T_k} v_i(A_i^\star)^p + \sum_{k=\rho+1}^r n_k \left(\frac{m_k}{n_k}\right)^p}{\sum_{k=1}^\rho \sum_{i \in T_k} v_i(A_i^\star)^p + \lambda^p \sum_{k=\rho+1}^r n_k  } \right)^{1/p} \, . \label{align:pof-ratio}
    \end{align}

    We will split the rest of the analysis into three cases: (1) $p<0$, (2) $p>0$, and (3) $p=0$.

    \textbf{\underline{Case I}: $p<0$}
    
    Noting that the first term in the numerator and the denominator in~\eqref{align:pof-ratio} is the same, to simplify this further, we will use \Cref{prop:fraction-inequality-neg}. The proposition is easily verified, and we skip a formal proof.

    \begin{proposition}
    Consider non-negative real numbers $x,y,a,b$ such that $x \geqslant y$, $b \geqslant a$, and $y+a>0$. Then for any fixed $p < 0$,
    \[
    \left(\frac{x + a}{x + b}\right)^{1/p} \leqslant \left( \frac{y + a}{y + b}\right)^{1/p}.
    \]
    \label{prop:fraction-inequality-neg}
    \end{proposition}

    In~\eqref{align:pof-ratio} we then let $x = \sum_{k=1}^\rho \sum_{i \in T_k} v_i(A_i^\star)^p$, $y = \lambda^p \sum_{k=1}^\rho n_k$, $a = \sum_{k=\rho+1}^r n_k \left(\frac{m_k}{n_k}\right)^p$, and $b = \lambda^p \sum_{k=\rho+1}^r n_k$. Then since $x \geqslant y$, $b \geqslant a$, and $y+a>0$, from \Cref{prop:fraction-inequality-neg} we get
    \begin{align*}
        \frac{\mathcal{W}_p(A^\star)}{\mathcal{W}_p(B)} & \leqslant ~  \displaystyle \left( \frac{\lambda^p \sum_{k=1}^\rho n_k + \sum_{k=\rho+1}^r n_k \left(\frac{m_k}{n_k}\right)^p}{n \lambda^p } \right)^{1/p} ~ = ~  \left( \frac{\sum_{k=1}^\rho n_k}{n} + \frac{\sum_{k=\rho+1}^r n_k \left(\frac{m_k}{n_k}\right)^p}{n \lambda^p } \right)^{1/p}.
    \end{align*}
    
    We define $\alpha := \sum_{k=1}^\rho n_k/n$, i.e., the ratio of number of types that retain their values in $B$. Replacing in the above expression, and using that $\lambda \geqslant W / \sum_{i=1}^\rho n_i$ from~\eqref{align:lambda-sum-2},
    \begin{align*}
        \frac{\mathcal{W}_p(A^\star)}{\mathcal{W}_p(B)} & \leqslant ~ \left( \alpha + \frac{\sum_{k=\rho+1}^r n_k \left(\frac{m_k}{n_k}\right)^p}{n W^p / \left(\sum_{i=1}^\rho n_i\right)^p} \right)^{1/p}.
    \end{align*}
    For each type $k$, $m_k \leqslant W$, since for agents of each type at most $W$ goods have positive value. Hence
    \begin{align*}
        \frac{\mathcal{W}_p(A^\star)}{\mathcal{W}_p(B)} & \leqslant ~ \left( \alpha + \frac{\sum_{k=\rho+1}^r n_k \left(\frac{W}{n_k}\right)^p}{n W^p / \left(\sum_{i=1}^\rho n_i\right)^p} \right)^{1/p} ~ = ~ \left( \alpha + \frac{\sum_{k=\rho+1}^r n_k^{1-p}}{n / \left(\sum_{i=1}^\rho n_i\right)^p} \right)^{1/p} \\  
            & = \left( \alpha + \alpha^p \sum_{k=\rho+1}^r \left(\frac{n_k}{n}\right)^{1-p}\right)^{1/p} \, .
    \end{align*}
    We now use~\Cref{claim:averaging-values}, choosing $x_k = n_k/n$, which gives us
    \begin{align*}
        \frac{\mathcal{W}_p(A^\star)}{\mathcal{W}_p(B)} & \leqslant  \left( \alpha + \alpha^p (r -\rho)  \left(\frac{\sum_{k=\rho+1}^r n_k}{n (r-\rho) }\right)^{1-p}\right)^{1/p} \, \\
        & = \left( \alpha + \alpha^p (r -\rho)^p  \left(\frac{n - \sum_{k=1}^\rho n_k}{n}\right)^{1-p}\right)^{1/p} ~ = ~ \left( \alpha + \alpha^p (r -\rho)^p  (1-\alpha)^{1-p}\right)^{1/p}.
    \end{align*}
    
    Finally, since $\rho \geqslant 1$, $r- \rho \leqslant s$ (where we defined $s = r-1$), hence we get the claim.

    \textbf{\underline{Case II}: $p>0$}
    
    Noting that the first term in the numerator and the denominator in~\eqref{align:pof-ratio} is the same, to simplify this further, we will use \Cref{prop:fraction-inequality}. The proposition is easily verified, and we skip a formal proof.
    
    \begin{proposition}
    Consider non-negative real numbers $x,y,a,b$ such that $x \geqslant y$, $a \geqslant b$ and $y+b>0$. Then for any fixed $p > 0$,
    \[
    \left(\frac{x + a}{x + b}\right)^{1/p} \leqslant \left( \frac{y + a}{y + b}\right)^{1/p}.
    \]
    \label{prop:fraction-inequality}
    \end{proposition}

    Observe that for any agent $i \in[n]$ such that $v_i(A^*_i) > 0$, we have that $\lambda \leqslant 2 \cdot v_i(A^*_i)$. Indeed, if $\lambda > 2 \cdot v_i(A^*_i)$, then from the discussion in \Cref{subsec:Upper-Bounds-PoE}, it follows that $2 \cdot v_i(A^*_i) < 1 + v_i(A^*_i)$, which, for integral valuations, implies that $v_i(A^*_i) = 0$.
    
    In~\eqref{align:pof-ratio} we then let $x = 2^p \sum_{k=1}^\rho \sum_{i \in T_k} v_i(A_i^\star)^p$, $y = \lambda^p \sum_{k=1}^\rho n_k$, $a = 2^p \sum_{k=\rho+1}^r n_k \left(\frac{m_k}{n_k}\right)^p$, and $b = 2^p \lambda^p \sum_{k=\rho+1}^r n_k$. Then since $x \geqslant y$, $a \geqslant b$, and $y+b>0$, from \Cref{prop:fraction-inequality} we get
    \begin{align}
        \displaystyle \left( \frac{\sum_{k=1}^\rho \sum_{i \in T_k} 2^p v_i(A_i^\star)^p + 2^p \sum_{k=\rho+1}^r n_k \left(\frac{m_k}{n_k}\right)^p}{\sum_{k=1}^\rho \sum_{i \in T_k} 2^p v_i(A_i^\star)^p + 2^p \lambda^p \sum_{k=\rho+1}^r n_k  } \right)^{1/p}
        & \leqslant &
        \displaystyle \left( \frac{\lambda^p \sum_{k=1}^\rho n_k + 2^p \sum_{k=\rho+1}^r n_k \left(\frac{m_k}{n_k}\right)^p}{\lambda^p \sum_{k=1}^\rho n_k + 2^p \lambda^p \sum_{k=\rho+1}^r n_k  } \right)^{1/p} \nonumber \\
        & \leqslant &
        \displaystyle \left( \frac{\lambda^p \sum_{k=1}^\rho n_k + 2^p \sum_{k=\rho+1}^r n_k \left(\frac{m_k}{n_k}\right)^p}{\lambda^p \sum_{k=1}^\rho n_k + \lambda^p \sum_{k=\rho+1}^r n_k  } \right)^{1/p} \nonumber \\
        & \leqslant &
        \displaystyle \left( \frac{\lambda^p \sum_{k=1}^\rho n_k + 2^p \sum_{k=\rho+1}^r n_k \left(\frac{m_k}{n_k}\right)^p}{n \lambda^p} \right)^{1/p} \, .
        \label{align:temp}
    \end{align}

The LHS in~\eqref{align:temp} is equal to the RHS in~\eqref{align:pof-ratio}. Thus, we get that
    \begin{align*}
        \frac{\mathcal{W}_p(A^\star)}{\mathcal{W}_p(B)} & \leqslant ~  \displaystyle \left( \frac{\lambda^p \sum_{k=1}^\rho n_k + 2^p \sum_{k=\rho+1}^r n_k \left(\frac{m_k}{n_k}\right)^p}{n \lambda^p } \right)^{1/p} ~ = ~  \left( \frac{\sum_{k=1}^\rho n_k}{n} + \frac{2^p \sum_{k=\rho+1}^r n_k \left(\frac{m_k}{n_k}\right)^p}{n \lambda^p } \right)^{1/p}.
    \end{align*}
    
    We define $\alpha := \sum_{k=1}^\rho n_k/n$, i.e., the ratio of number of types that retain their values in $B$. Replacing in the above expression, and using that $\lambda \geqslant W / \sum_{i=1}^\rho n_i$ from~\eqref{align:lambda-sum-2},
    \begin{align*}
        \frac{\mathcal{W}_p(A^\star)}{\mathcal{W}_p(B)} & \leqslant ~ \left( \alpha + \frac{2^p \sum_{k=\rho+1}^r n_k \left(\frac{m_k}{n_k}\right)^p}{n W^p / \left(\sum_{i=1}^\rho n_i\right)^p} \right)^{1/p}.
    \end{align*}
    For each type $k$, $m_k \leqslant W$, since for agents of each type at most $W$ goods have positive value. Hence
    \begin{align*}
        \frac{\mathcal{W}_p(A^\star)}{\mathcal{W}_p(B)} & \leqslant ~ \left( \alpha + \frac{2^p \sum_{k=\rho+1}^r n_k \left(\frac{W}{n_k}\right)^p}{n W^p / \left(\sum_{i=1}^\rho n_i\right)^p} \right)^{1/p} ~ = ~ \left( \alpha + \frac{2^p \sum_{k=\rho+1}^r n_k^{1-p}}{n / \left(\sum_{i=1}^\rho n_i\right)^p} \right)^{1/p} \\  
            & = \left( \alpha + 2^p \alpha^p \sum_{k=\rho+1}^r \left(\frac{n_k}{n}\right)^{1-p}\right)^{1/p} \, .
    \end{align*}
    We now use~\Cref{claim:averaging-values}, choosing $x_k = n_k/n$, which gives us
    \begin{align*}
        \frac{\mathcal{W}_p(A^\star)}{\mathcal{W}_p(B)} & \leqslant  \left( \alpha + 2^p \alpha^p (r -\rho)  \left(\frac{\sum_{k=\rho+1}^r n_k}{n (r-\rho) }\right)^{1-p}\right)^{1/p} \, \\
        & = \left( \alpha + 2^p \alpha^p (r -\rho)^p  \left(\frac{n - \sum_{k=1}^\rho n_k}{n}\right)^{1-p}\right)^{1/p} ~ = ~ \left( \alpha + 2^p \alpha^p (r -\rho)^p  (1-\alpha)^{1-p}\right)^{1/p}.
    \end{align*}
    
    Finally, since $\rho \geqslant 1$, $r- \rho \leqslant s$ (where we defined $s = r-1$), hence we get the claim.

    \textbf{\underline{Case III}: $p=0$}. In this case, the Nash welfare of an allocation is the geometric mean of the values of the agents. By definition of the truncated allocation $B$, agents of the first $\rho$ types have the same value in $A^*$ and $B$, hence

    \begin{align*}
        \frac{\mathcal{W}_0(A^\star)}{\mathcal{W}_0(B)} 
        & = \left( \frac{\Pi_{i=1}^{n} v_i(A^*_i)}{\Pi_{i=1}^{n} v_i(B_i)} \right)^{1/n} \\
        & = \left( \frac{ \Pi_{k=1}^{\rho} \Pi_{i \in T_k} v_i(A^*_i) \cdot \Pi_{k=\rho+1}^{r} \Pi_{i \in T_k} v_i(A^*_i)}{ \Pi_{k=1}^{\rho} \Pi_{i \in T_k} v_i(B_i) \cdot \Pi_{k=\rho+1}^{r} \Pi_{i \in T_k} v_i(B_i)} \right)^{1/n} \\
        & = \left( \frac{ \Pi_{k=\rho+1}^{r} \Pi_{i \in T_k} v_i(A^*_i)}{ \Pi_{k=\rho+1}^{r} \Pi_{i \in T_k} v_i(B_i)} \right)^{1/n}
    \end{align*}

    In $A^*$, by~\Cref{corollary:averaging-subset}, for a fixed type $k$, we can bound $\Pi_{i \in T_k} v_i(A^*_i)$ from above by $(m_k/n_k)^{n_k}$ $\le (W/n_k)^{n_k}$. Further, each agent of type $> \rho$ has $v_i(B_i) = \lambda$, and from~\eqref{align:lambda-sum-2}, $\lambda \geq W/\sum_{i=1}^\rho n_i$.

    Let $n' := \sum_{i=1}^\rho n_i$ be the number of agents of the first $\rho$ types. Then substituting these values, we get

    \begin{align*}
        \frac{\mathcal{W}_0(A^\star)}{\mathcal{W}_0(B)} 
        & \le \left( \frac{ \Pi_{k=\rho+1}^{r} (W/n_k)^{n_k}}{ \Pi_{k=\rho+1}^{r} (W/n')^{n_k}} \right)^{1/n} \\
        & = \left( \frac{ (n')^{n-n'}}{ \Pi_{k=\rho+1}^{r} (n_k)^{n_k}} \right)^{1/n} \\
    \end{align*}

    Noting that $\sum_{k=\rho+1}^{r} n_k = n-n'$, and each $n_k \ge 1$, the product $\Pi_{k=\rho+1}^{r} (n_k)^{n_k}$ is maximized when the $n_k$'s are equal, hence each $n_k = (n-n')/(r-\rho)$. With this substitution,

    \begin{align*}
        \frac{\mathcal{W}_0(A^\star)}{\mathcal{W}_0(B)} 
        & \le \left( \frac{n'}{ (n-n')/(r-\rho)} \right)^{(n-n')/n} \\
    \end{align*}

    Recalling that $\alpha = n'/n$, and further $s = r-1 \ge r- \rho$,
   \begin{align*}
        \frac{\mathcal{W}_0(A^\star)}{\mathcal{W}_0(B)} 
        & \le \left( \frac{s \, \alpha}{ (1-\alpha)} \right)^{1-\alpha} \\
    \end{align*}

    \noindent which is the required expression.
\end{proof}


We are now ready to present our upper bounds.

\pofupperbounds*

\begin{proof}
    Our starting point is~\Cref{lem:pof-ub-base}. For $p \rightarrow 0$, the PoE is at most $\sup_{\alpha \in [0,1]} (s \alpha/(1-\alpha))^{(1-\alpha)}$. Let $\beta := \alpha/(1-\alpha)$, then $1-\alpha = 1/(1+\beta)$, and hence the upper bound on the PoE is $\sup_{\beta \ge 0} (s\beta)^{1/(\beta+1)}$. 

    Some calculus shows that the maximum value of this function is $\exp{(W(s/e))}$, where $W(\cdot)$ is the Lambert W function, which is the inverse of the function $f(x) = xe^x$. Further, $W(x) \leqslant \ln x - \ln \ln x + \frac{e \ln \ln x }{(e-1) \ln x}$ for $x \geqslant e$. The last term $\frac{e \ln \ln x }{(e-1) \ln x}$ $\leqslant 1$ for $x \geqslant e$. Hence for $s \geqslant e^2$, we get that the PoE is bounded by 
    \[
    \exp{(W(s/e))} \leqslant \exp{(\ln (s/e) - \ln \ln (s/e) + 1)} = \frac{s}{\ln (s/e)}
    \]

    \noindent giving the required bound on the PoE.

    For $p \in (0,1)$, again from~\Cref{lem:pof-ub-base}, the upper bound on the PoE can be written as
    
    \begin{align*}
        \sup_{\alpha \in [0,1]} \left(\alpha \times 1 + (1-\alpha) \times (2s\alpha/(1-\alpha))^p\right)^{1/p}
    \end{align*}

    Since $p \in (0,1)$, $f(x) = x^{1/p}$ is a convex function, and hence by Jensen's inequality this is at most

    \begin{align*}
        \sup_{\alpha \in [0,1]} \left( \alpha \times 1^{1/p} + (1-\alpha) \times (2s\alpha/(1-\alpha))\right) ~ = ~ \sup_{\alpha \in [0,1]} \left( \alpha + 2s\alpha\right) ~ = ~ 1 + 2s
    \end{align*}

    which is the upper bound claimed, for $p \in (0,1)$.





For $p < 0$, we separate the two cases $\alpha \geqslant 1/2$ and $\alpha \leqslant 1/2$. If $\alpha \geqslant 1/2$, then the expression from~\Cref{lem:pof-ub-base} evaluates to
    \begin{align}
    \left(\alpha + \alpha^p s^p (1-\alpha)^{1-p}\right)^{1/p} & \leqslant \alpha^{1/p} \leqslant (1/2)^{1/p} \, \, . \label{eqn:alpha-ge-half}
    \end{align}

\noindent If $\alpha \leqslant 1/2$, then $(1-\alpha) \geqslant 1/2$, and hence,
    \begin{align*}
    \left(\alpha + \alpha^p s^p (1-\alpha)^{1-p}\right)^{1/p} & \leqslant \left(\alpha + \alpha^p s^p 2^{-1+p}\right)^{1/p} \, .
    \end{align*}

Let $z := \alpha + \alpha^p (2s)^p /2$ be the parenthesized expression; our goal is to minimize this (since the exponent $1/p$ is negative, this will give us an upper bound on the PoE). Differentiating w.r.t. $\alpha$ gives us
\begin{align*}
    \frac{d z}{d \alpha} & = 1 + \frac{p}{2} (2s)^p \alpha^{p-1} \, .
\end{align*}

Since $p$ is negative, this increases with $\alpha$, and hence the derivative is a convex function with a unique minima, obtained at
\[
\displaystyle \alpha^* = \frac{(-2/p)^{1/(p-1)}}{(2s)^{p/(p-1)}}  
\]

\noindent or $(\alpha^* 2s)^p = -2 \alpha^*/p$. Replacing this value gives us
\begin{align*}
    \left(\alpha + \alpha^p s^p (1-\alpha)^{1-p}\right)^{1/p} & \leqslant \alpha^{*^{1/p}} (1 - 1/p)^{1/p}
\end{align*}

\noindent For $p < 0$ , $1-1/p \geqslant 1$, and hence $(1 - 1/p)^{1/p} \leqslant 1$. Hence
\begin{align*}
    \left(\alpha + \alpha^p s^p (1-\alpha)^{1-p}\right)^{1/p}  \leqslant \alpha^{*^{1/p}} & = (2s)^{1/(1-p)} (-2/p)^{1/p(p-1)} \\
        & = s^{1/(1-p)} 2^{-1/p} (-1/p)^{1/p(p-1)} \, . 
\end{align*}

This is greater than $2^{-1/p}$, the expression we obtain for $\alpha \geqslant 1/2$ in~\Cref{eqn:alpha-ge-half}, and hence this is a bound on the PoE for $p < 0$.

Finally for $p \leqslant -1$, let us consider the coefficient of $s^{1/(1-p)}$ obtained previously, namely $2^{-1/p} (-1/p)^{1/p(p-1)}$. This is an increasing function of $p$, and hence the maximum value obtained is $2$, at $p = -1$. Hence for $p \leqslant -1$, the PoE is at most $2s^{1/(1-p)}$.
\end{proof}

\section{PoE bounds for Doubly Normalised Instances}
\label{sec:doubly-normalised}

So far, we have considered instances with binary additive normalised valuations, where each agent values the same number $W$ of goods. In this case, for the utilitarian welfare, we have seen that the PoE can be as bad as $r$, the number of types of agents. In this section, we consider instances with further structure. In \emph{doubly normalised} instances, each good $g$ is valued by the same number $W_c$ of agents. Thus, $v_i(M) = W$ for all $i \in N$, and additionally, $\sum_{i \in N} v_i(g) = W_c$ for every good $g \in M$. The valuation matrix $V$ is thus both row and column normalised. Such instances are intuitively ``balanced,'' and we ask if this balance is reflected in the PoE for such instances. This indeed turns out to be the case. 

\doublynormalised*

For an undirected graph, the edge-incident matrix $X$ has entry $X_{i,e} = 1$ if edge $e$ is incident on vertex $i$, and $X_{i,e} = 0$ otherwise. We will use the following well-known property of edge-incidence matrices for bipartite graphs.

\begin{proposition}[e.g.,~\citealp{Schrijver98}]
    If $G$ is a bipartite graph, then the edge-incidence matrix of $G$ is totally unimodular.
    \label{prop:bipartite-TU}
\end{proposition}

\begin{proof}[Proof of~\Cref{thm:doubly-normalised}]
Let $V$ be the valuation matrix  for a doubly normalised instance, where each row sums to $W$ and each column sums to $W_c$. Divide each entry by $W_c$. Let $V^f$ be the resulting matrix. Then $V^f$ satisfies: (i) each entry is either $0$ or $1/W_c$, (ii) each column sums to 1, and (iii) each row sums to $W/W_c$. We will show that the matrix $V^f$ can be represented as the convex combination of nonnegative integer matrices $X^1$, $\ldots$, $X^t$ so that for any matrix $X^k$ in this decomposition, each column sums to $1$ and each row sums to either $\lceil W/W_c \rceil$ or $\lfloor W/W_c \rfloor$. Assuming such a decomposition, fix any such matrix $X^k$ in this decomposition. Clearly, due to (ii) and  nonnegativity, each entry of $X^k$ is either $1$ or $0$. Further if the entry $X^k_{i,g} = 1$, then $V^f_{i,g} = 1/W_c$ since $V^f$ is a convex combination of the $M$-matrices, and hence $V_{i,g} = 1$,

Consider then the allocation $A$ that assigns good $g$ to agent $i$ if $X^k_{i,g} = 1$. In this allocation, following the properties of $X^k$, each good is assigned to an agent that has value 1 for it, and each agent is assigned either $\lceil W/W_c \rceil$ or $\lfloor W/W_c \rfloor$ goods. The allocation is thus \EQ{1} and maximizes the utilitarian welfare. Further by~\Cref{prop:leximin-if-nearly-equal} this is also a leximin allocation, and hence by~\Cref{prop:NashOpt_Simultaneously_Maximizes_p-mean} and~\Cref{prop:nsw-is-optimal} this maximizes the $p$-mean welfare for all $p \le 1$, proving the theorem.

It remains to show that $V$ can be decomposed as stated. To see this, consider a complete bipartite graph $G = (A \cup B, E)$ with $|A| = n$ and $|B| = m$. To each edge $\{i,g\}$ with $i \in A$, $g \in B$, we associate a variable $x_{ig}$. Consider now the set of linear constraints:
\[
\begin{array}{lrl}
\forall i \in A, & \sum_{g \in B} x_{ig} & \ge \lfloor W/W_c \rfloor \\
\forall i \in A, & \sum_{g \in B} x_{ig} & \le \lceil W/W_c \rceil \\
\forall g \in B, & \sum_{i \in A} x_{ig} & = 1 \\
\forall i \in A, \, g \in B, & x_{ig} & \ge 0
\end{array}
\]

Together, these linear constraints ask for a fractional set of edges that have degree $1$ for each vertex in $B$ and degree between $\lfloor W/W_c \rfloor$ and $\lceil W/W_c \rceil$ for each vertex in $A$.

Consider the polytope obtained by these inequalities. Taking $x_{ig} = V^f_{ig}$ satisfies these constraints. Further, it can be seen that the constraint matrix is equal to the edge-incidence matrix for the bipartite graph $G$ (with the rows corresponding to vertices $i \in A$ repeated, and the identity matrix appended for nonnegativity of the variables). Hence, the constraint matrix is totally unimodular by~\Cref{prop:bipartite-TU}, and thus the extreme points of the polytope are integral. Since $V^f$ is a point in the polytope, $V^f$ can be represented as the convex combination of nonnegative integral matrices $X^1$, $\ldots$, $X^t$ corresponding to the vertices of the polytope, as required.
\end{proof}

We make two remarks. Firstly, note that since each matrix $X^k$ in the convex decomposition of $V^f$ gives us an \EQ{1} allocation with maximum utilitarian welfare, the convex combination gives us a randomized allocation that is ex ante \EQ{}, and ex post \EQ{1} and welfare optimal. Secondly, the doubly normalised constraint is sufficient, but not necessary, for the price of equity to be 1. Consider an instance with $3$ agents $\{a_1, a_2, a_3\}$ and $4$ goods $\{g_1, g_2, g_3, g_4\}$ such that $a_1$ values $\{g_1, g_2\}$ while $a_2$ and $a_3$ both value $\{g_3, g_4\}$. This instance is not column normalised, but admits an $\EQ1{}$ allocation with optimal welfare.

In the appendix, we offer an alternate proof of \Cref{thm:doubly-normalised}, based on a so-called ``eating argument'' and an extension of Hall's theorem.

\section{PoE Bounds for Binary Submodular Valuations}
\label{sec:binary-submodular}

We now consider the more general case of binary submodular valuations. Here we focus on the utilitarian welfare, and show that our results for binary additive valuations that bound the PoE by the number of types of agents do not extend to binary submodular valuations. We first show that from prior work (see \Cref{prop:mrv-ef1} below), it follows that if the agents have identical valuations, then PoE is 1 for the $p$-mean welfare objective for all $p \le 1$.

\mrvpofonetype*

As earlier, an allocation $A = (A_1,\dots,A_n)$ is \emph{clean} if for all agents $i$, $v_i(A_i) = |A_i|$, that is, no good is wastefully allocated. We note that, given any allocation $A$, we can obtain a clean (possibly partial) allocation $\hat{A}$ so that $v_i(A_i) = v_i(\hat{A}_i)$ for all agents $i$ by repeatedly removing wasted items from the allocation $A$. We will use the following result due to \cite{BCI+21finding}. 


\begin{proposition}[\citealp{BCI+21finding}, Corollary 3.8]
For binary submodular valuations, any clean, utilitarian optimal (partial) allocation that minimizes $\Phi(A):= \sum_i v_i(A_i)^2$ among all utilitarian optimal allocations is \EF{1}.
\label{prop:mrv-ef1}
\end{proposition}

\begin{proof}[Proof of~\Cref{prop:mrv-pof-1}]
Let $A^\star$ be a Nash welfare maximizing allocation for the given instance. 
%
We will show that under identical binary submodular valuations, $A^*$ can be transformed into an \EQ{1} allocation without any change in the Nash welfare objective, thus implying that PoE is 1 for Nash welfare. Furthermore, from \Cref{prop:NashOpt_Simultaneously_Maximizes_p-mean}, we know that any allocation that maximizes Nash welfare also simultaneously maximizes $p$-mean welfare for all $p \leqslant 1$. This would imply that PoE is 1 for $p$-mean welfare objective for all $p \leqslant 1$.

First, we will transform $A^*$ into a clean partial allocation via the following procedure: For each agent $i$ with $v_i(A_i^*) > |A_i^*|$, there must be a wasted good in $A_i^*$; we simply remove such wasted goods until we get a clean partial allocation $\hat{A}$. Next, we will add back the removed goods \emph{arbitrarily} to obtain a complete allocation $A$ (in particular, adding back the removed goods may get back the original allocation $A^*$).

Note that $\hat{A}$ is a partial allocation with $v_i(\hat{A}_i) = v_i(A_i^*)$ for each agent $i$; in other words, $A^*$ and $\hat{A}$ have the same $p$-mean welfare for all $p \leqslant 1$. We know from \Cref{prop:NashOpt_Simultaneously_Maximizes_p-mean} that, for all $p \leqslant 1$, $A^*$ maximizes the $p$-mean welfare. The same holds true for $\hat{A}$.

By adding the removed goods back to $\hat{A}$, the utility of any agent cannot decrease; that is, for every agent $i$, $v_i(A_i) = v_i(\hat{A}_i)$. This means that $A$ is a complete allocation that simultaneously maximizes the $p$-mean welfare for all $p \leqslant 1$.

By \Cref{prop:nsw-is-optimal}, allocation $A$ minimizes the strictly convex function $\Phi(A):= \sum_i v_i(A_i)^2$ among all utilitarian allocations, and the same holds for the partial allocation $\hat{A}$. Then, by~\Cref{prop:mrv-ef1}, we get that $\hat{A}$ is \EF{1}. By the identical valuations assumption, $\hat{A}$ is also \EQ{1}.

In going from $\hat{A}$ to $A$, each good that is added back has zero marginal value for the agent it is assigned to under $A$. Thus, the allocation $A$ is also \EQ{1}, which readily implies that for all $p \leqslant 1$, the \PoE{} for $p$-mean welfare is 1, as desired.
\end{proof}

The bound in \Cref{prop:mrv-pof-1} is, in a certain sense, the best that can be obtained. We will now show that with more than one type of agent under binary submodular valuations, the PoE is at least $n/6$ for utilitarian welfare. Hence we cannot obtain bounds on the PoE that depend on the number of agent types for all $p \le 1$, as we did for binary additive valuations.
    
\pofmrvtwotypes*

\begin{proof}
In our example for the lower bound, we represent goods as vectors (i.e., elements of a linear matroid). Then the value of an agent for a bundle is just the number of linearly independent vectors in the bundle. Fix $k \in \mathbb{N}$. Our example will have $2k$ agents and $k^2+k$ goods.\\

\noindent \textbf{Goods:} There are $k(k+1)$ goods, consisting of $k+1$ groups of $k$ goods each. The groups are $G_1$, $G_2$, $\ldots$, $G_{k+1}$.

\noindent \textbf{Agents:} There are $2k$ agents, with $k$ agents of type 1 and $k$ agents of type 2. Agents of type $1$ see goods in $G_1$ as the standard basis vectors for $\mathbb{R}^k$, and goods in $G_j$ for $j \neq 1$ as zero vectors. Thus, for an agent $i$ of type 1, $v_i(G_1) = v_i(M) = k$, and $v_i(G_j) = 0$ for $j > 1$.

Agents of type $2$ see the goods in each group $G_i$ as the standard basis vectors for $\mathbb{R}^k$, and hence for an agent $i$ of type 2, $v_i(G_j) = v_i(M) = k$, for all $j \in [k+1]$. Thus, the valuations are normalised.

In an \EQ{1} allocation, each agent of type 1 has value at most 1, and hence the social welfare is at most $3k$. In the optimal allocation, each agent of type 1 gets a single vector from $G_1$. Each agent of group 2 gets assigned an entire group $G_j$ of vectors, and hence has value $k$. The optimal social welfare is thus $k + k^2$, and hence the PoE is at least $k/3$, or $n/6$, where $n$ is the number of agents.
\end{proof}

Note that for the example in the proof of \Cref{thm:pofmrvtwotypes}, for any $p \in (0,1]$, the PoE is 

\[
\Lambda_p ~=~ \displaystyle \left(\frac{\frac{1}{2k}\left(k \times 1 + k \times k^p\right)}{\frac{1}{2k}\left(k \times 1 + k \times 2^p\right)} \right)^{1/p} ~=~ \left(\frac{1 + k^p}{ 1 + 2^p} \right)^{1/p} ~\geqslant~ \frac{k}{3^{1/p}} \, ,
\]

\noindent and hence the PoE depends on the number of agents, even with two types. Similarly, for the Nash social welfare, one obtains the PoE as $\sqrt{k/2} = \sqrt{n/4}$.

For $p < 0$, for this example, the PoE is a constant that depends on $p$ (for example, for $p = -1$, the PoE for this example is $1.5$). It is possible that for $p < 0$ the PoE may depend on the agent types, rather than number of agents. We leave this as an open question.

Despite this, we show that $2n$ is an upper bound on the PoE for all $p \le 1$. For an allocation $A = (A_1, \ldots, A_n)$ of the goods, we say good $g$ is \emph{valuable} for $i$ if $v_i(A_i \cup g) > v_i(A_i)$ (and $i$ \emph{values} $g$ in this case). 




    
\pofmrvubn*

\begin{proof}
As before, let $A^\star$ be an allocation with optimal Nash welfare. If $A^\star$ is an \EQ{1} allocation, we are done, since from~\Cref{prop:NashOpt_Simultaneously_Maximizes_p-mean}, $A^\star$ simultaneously maximizes the $p$-mean welfare for all $p \leqslant 1$. Else, we construct the truncated allocation $B$ as described in~\Cref{subsec:Truncated_Allocation}. We will show that for every agent $i$ with non-zero value in $B$, $v_i(B_i) \ge W/(2n)$, where $W$ is the normalisation constant. It follows that the PoE is bounded by $2n$ for all $p \leqslant 1$.

Consider the allocation $A^\star$. Let $i_l$ be a minimum positive value agent in $A^\star$. Note that $A_{i_l}^* = B_{i_l}$. Let $\nu$ be the value of agent $i_l$ under $A^*$. Since $v_{i_l}(M) = W$, there are $W - \nu$ goods that $i_l$ values that are allocated to other agents. Further, any agent $i \neq i_l$ is allocated at most $\nu+1$ goods that $i_l$ values, since otherwise, we can transfer a good that $i_l$ values from $i$ to $i_l$ and increase the Nash social welfare of allocation $A^*$. Hence, $W-\nu \le (n-1)(\nu+1)$, or $W \le n\nu + n - 1 \le 2n\nu$ for $\nu \ge 1$. Thus for any agent $i$, $v_i(B_i) \ge v_i(B_{i_l}) = \nu \ge W/(2n)$, as required. 
%
%
%
%
\end{proof}

\section{Some Concluding Remarks on Chores}

Our focus in the paper has been on goods, where agents have non-negative marginal utility for all items. We briefly remark on the case of bads or chores, where all marginal utilities are non-positive. Consider any instance with binary additive valuations, i.e., the value of each item is either $0$ or $-1$. It is not hard to see that in these instances, there is always a utilitarian optimal \EQ{1} allocation: if chore $c$ has value $0$ for an agent $i$, assign $c$ to $i$. The remaining chores have value $-1$ for all agents, and can be assigned using the round robin procedure. This allocation is clearly \EQ{1} and also achieves the best possible utilitarian welfare.

For more general additive instances with chores, we now show that the PoE is unbounded, even in very simple cases.\footnote{For chores, we adopt the natural definition of PoE: the ratio of the utilitarian welfare of the best \EQ{1} allocation, to the maximum utilitarian welfare obtainable in any allocation. Note that if the denominator is $0$, then so is the numerator (and this can be identified in polynomial time).} To this end, consider the following example involving $2n$ items and $n+1$ agents. The first $n$ agents mildly dislike the first $n$ chores and severely dislike the last $n$, while it is the opposite for the $(n+1)$th agent, who strongly dislikes the first $n$ items and mildly dislikes the last $n$. 

\begin{center}
\begin{tabular}{|c|c|c|}
\hline
& $c_1,\cdots,c_n$ & $c_{n+1},\cdots,c_{2n}$ \\
\hline
$a_1$, $\ldots$, $a_n$ & $-\epsilon$ & $-1$ \\
\hline
$a_{n+1}$ & $-1$ & $-\epsilon$ \\
\hline
\end{tabular}
\vspace{0.2in}
\end{center}

In this example, the maximum utility is $-2n \epsilon$: assign the first $n$ chores to the first agent and the last $n$ chores to the last agent. On the other hand, in any \EQ{1} allocation, the last agent can get at most 2 chores, and hence some agent gets a chore that they value at $-1$. The PoE is thus at least $1/(2n\epsilon)$, which can be made arbitrarily large by choosing $\epsilon$ appropriately. Note that this instance has two item types, two agent types, and only two distinct entries in the valuation matrix. Relaxing any of these conditions implies identical valuations, where the PoE is $1$; so, in some sense, this is a ``minimally complex'' example that already exhibits unbounded PoE. There is thus a sharp change in the PoE between instances where the values are in $\{0,-1\}$ and those where the values are in $\{-\epsilon,-1\}$. While the PoE is unbounded as $\epsilon$ approaches $0$, it ``snaps back'' to $1$ at $\epsilon=0$. 

To conclude, we obtain nearly tight bounds on the price of equity in terms of agent types for the $p$-mean welfare spectrum. This captures, as special cases, the notions of utilitarian, egalitarian, and Nash welfare. Our bounds are in terms of agent types ($r$) rather than the number of agents. Overall, our results provide a fine-grained perspective on the behavior of the price of equity parameterized by $p$ and $r$. 

In future work, it would be interesting to extend the insights that we obtain in this work beyond the domain of binary valuations. We also propose obtaining bounds on the PoE parameterized by other structural parameters, such as the number of item types. We note that for additive valuations, the rank of the valuation matrix is a lower bound on the number of item types, and hence~\Cref{thm:ubrank} bounds the PoE in this case by the number of item types as well.

\section*{Acknowledgments}
We thank the anonymous reviewers from SAGT 2023 and COMSOC 2023 for helpful comments. We are especially grateful to an anonymous reviewer from COMSOC 2023 who pointed out an error in the proof of~\Cref{lem:pof-ub-base} in an earlier version of the draft and provided several other helpful comments. RV acknowledges support from DST INSPIRE grant no. DST/INSPIRE/04/2020/000107 and SERB grant no. CRG/2022/002621.


\bibliographystyle{plainnat} 
\bibliography{References}

\newpage

\section{Appendix}

\subsection{A proof for~\Cref{prop:NashOpt_Simultaneously_Maximizes_p-mean}}
\label{subsec:NashOpt}

\NashOptSimultaneouslyMaximizespmean*

The following property of leximin allocations is useful in the proof.

\begin{proposition}
    For agents with binary submodular valuations, let $A$ be a utilitarian optimal allocation so that $\max_i v_i(A_i) \le \min_i v_i(A_i) + 1$. Then $A$ is a leximin allocation.
    \label{prop:leximin-if-nearly-equal}
\end{proposition}

\begin{proof} (of \Cref{prop:leximin-if-nearly-equal})
    Assume without loss of generality that $v_i(A_i) \le v_{i+1}(A_{i+1})$ for $i \in [n-1]$. If $A$ is not a leximin allocation, let $A'$ be a leximin allocation. Let permutation $\pi \in S_n$ be such that $v_{\pi(i)}(A'_{\pi(i)}) \le v_{\pi(i+1)}(A'_{\pi(i+1)})$ for $i \in [n-1]$. Then there exists $k \in [n]$, so that $v_i(A_i) = v_{\pi(i)}(A'_{\pi(i)})$ for $i < k$, and $v_k(A_k) < v_{\pi(k)}(A'_{\pi(k)})$. Note that for $i > k$, 
    \[
     v_i(A_i) ~\le~ v_k(A_k) + 1 ~\le~ v_{\pi(k)}(A'_{\pi(k)}) ~\le~ v_{\pi(i)}(A'_{\pi(i)}) \, .
    \]

     But then $\sum_{i=1}^n v_{\pi(i)}(A'_{\pi(i)}) > \sum_{i=1}^n v_i(A_i)$, and allocation $A$ cannot be utilitarian optimal.
\end{proof}

The following results are shown by~\cite{BCI+21finding}. 

\begin{proposition}[\citealp{BCI+21finding}, Lemma 3.12]
   For agents with binary submodular valuations, let $A$ be a utilitarian optimal allocation that is not a leximin allocation. Let agents $i$, $j$ be such that $v_j(A_j) \ge v_i(A_i) + 2$.\footnote{By~\Cref{prop:leximin-if-nearly-equal}, such agents must exist.} Then there is another allocation $A'$ that is utilitarian optimal and satisfies (i) $v_j(A_j') = v_j(A_j) - 1$, (ii) $v_i(A_i') = v_i(A_i) + 1$, and (iii) the values for other agents are unchanged.
   \label{prop:not-leximin-1}
\end{proposition}

Note that in the above proposition, the allocation $A'$ is a lexicographic improvement on $A$. 

\begin{proposition}[\citealp{BCI+21finding}, Lemma 3.13]
Let $\Psi$ be a symmetric concave function, and $A$ be a utilitarian optimal allocation with agents $i$, $j$ such that $v_j(A_j) \ge v_i(A_i) + 2$. Let $A'$ be another utilitarian optimal allocation that satisfies (i) $v_j(A_j') = v_j(A_j) - 1$, (ii) $v_i(A_i') = v_i(A_i) + 1$, and (iii) the values for other agents are unchanged. Then $\Psi(A') \ge \Psi(A)$.
\label{prop:not-leximin-2}
\end{proposition}

We can now prove~\Cref{prop:NashOpt_Simultaneously_Maximizes_p-mean}. Firstly, note that the Nash welfare maximizing allocation is also leximin from~\Cref{prop:nsw-is-optimal}, and hence maximizes the egalitarian welfare (in other words, maximizes $p$-mean welfare for $p \rightarrow -\infty$). For any fixed $p \le 1$, let $A$ be an allocation that maximizes the $p$-mean welfare. Since the $p$-mean welfare is strictly increasing, allocation $A$ is Pareto optimal, and hence from~\Cref{prop:po-means-utilitarian-optimal} is also utilitarian optimal. We will show that there exists an allocation $B$ so that $\mathcal{W}_p(B) = \mathcal{W}_p(A)$, and $B$ is a leximin allocation. If $A$ is not leximin, then by Propositions~\ref{prop:not-leximin-1} and~\ref{prop:not-leximin-2}, there is an allocation $A'$ so that $\mathcal{W}_p(A') \ge \mathcal{W}_p(A)$, and $A'$ lexicographically dominates $A$. Then either $A'$ is a leximin allocation, or we can continue in this manner until we get a leximin allocation $B$ with $\mathcal{W}_p(B) \ge \mathcal{W}_p(A)$, as required. Finally, by~\Cref{prop:po-means-utilitarian-optimal}, an allocation is leximin if and only if it maximizes the Nash social welfare, hence if allocation $A^*$ maximizes the Nash social welfare, it also maximizes the $p$-mean welfare for all $p \le 1$.

\subsection{An alternate proof of \Cref{thm:doubly-normalised}}

We now turn to an alternate proof of the fact that doubly normalised instances have PoE 1.

We first state some results that we will use. Consider a bipartite graph $G$ with bipartition $A$ and $B$. A set of edges $T \subseteq E(G)$ is said to be a \emph{$q$-expansion} from $A$ to $B$ if every vertex of $A$ is incident to exactly $q$ edges in $T$ and exactly $q|A|$ vertices in $B$ are incident on $T$. A perfect matching, for instance, is a $1$-expansion, and a star with $q$ leaves is a $q$-expansion.  

\begin{lemma}
[\citealp{DBLP:books/sp/CyganFKLMPPS15}, Lemma 2.17] 
\label{lem:qexpansion}
Let $G$ be a bipartite graph with bipartition $A$ and $B$. There is a $q$-expansion from $A$ to $B$ if and only if $|N(X)| \geq q|X|$ for every $X \subseteq A$. Furthermore, if there is no $q$-expansion from $A$ to $B$, then a set $X \subseteq A$ such that $|N(X)| < q|X|$ can be found in polynomial time.
\end{lemma}

A non-negative, square ($m \times m$) matrix $Y$ is said to be \emph{doubly stochastic} if the sum of entries in each row and each column is $1$, that is, $\sum_{i \in [m]} Y[i][j] = 1$ $\forall \, j \in [m]$ and  $\sum_{j \in [m]} Y[i][j] = 1$ $\forall \, i \in [m]$. A \emph{permutation matrix} is a doubly stochastic matrix such that all the entries are either $0$ or $1$.
The following result, known as the Birkhoff-von Neumann Theorem, states that a doubly stochastic matrix can be represented as a convex combination of permutation matrices.

\begin{theorem}[\citealp{1573387450959988992,Neumann+1953+5+12}] 
\label{thm:birkhoff} 
Let $Y$ be a doubly stochastic matrix. Then, there exist positive weights $w_1, w_2, \ldots w_k$ and permutation matrices $P_1, P_2, \ldots P_k$ such that $\sum_{i \in [k]} w_i = 1$ and $Y = \sum_{i \in [k]} w_i P_i$.

In other words, the convex hull of the set of all permutation matrices is the set of doubly-stochastic matrices.
\end{theorem}

We are now ready to prove that doubly normalised instances have PoE $1$.

\doublynormalised*

\begin{proof} (of \Cref{thm:doubly-normalised}) Let $\mathcal{I}= \langle N, M, \mathcal{V} \rangle$ be a doubly normalised instance. Let $G = \left(N, M\right)$ be the corresponding ($W, W_c$)-regular bipartite graph with agents and goods as bi-partitions and $\left(i, g\right) \in E(G)$ if and only if agent $i$ values the good $g$. Note that the number of edges in $G$ is $nW$, as exactly $W$ edges are incident on each agent. Likewise, as exactly $W_c$ edges are incident on each of the $m$ goods, therefore, $|E(G)| = nW = mW_c$. 

We first consider the case when $n/m = W/W_c = p$ for some integer $p$. That is, the number of agents is an integer multiple of the number of goods, and show the existence of a non-wasteful $\EQ{}$ allocation that allocates a utility of $p$ to every agent.

Consider any subset $S \subseteq N$. The number of edges from $S$ to its neighborhood $N(S)$ is exactly $W|S|$. The number of edges incident on $N(S)$ in $G$ is exactly $W_c|N(S)|$. Then $W_c |N(S)| \ge W |S|$, and hence $|N(S)| \ge p |S|$. 

By \Cref{lem:qexpansion}, $G$ must have a $p$-expansion, say $T$, from $N$ to $M$. Now consider the allocation $A$ that allocates along the edges of $T$. Precisely, if $(i, g) \in T$, then good $g$ is allocated to the agent $i$ under $A$. Then by definition of $p$-expansion, $A$ is an indeed an \EQ{} allocation as it allocates exactly $p$ goods to every agent. Since $A$ is non-wasteful, it achieves the optimal utilitarian welfare $m$.
    
    
    
    Now suppose $W$ is not an integer multiple of $W_c$. We propose the following version of the probabilistic serial algorithm that constructs a non-wasteful \EQ1{} allocation.\footnote{The probabilistic serial algorithm was proposed by~\cite{BM01new} in the context of the assignment problem where the number of goods and agents is the same. Subsequently,~\cite{AFS+23best} used this algorithm (in combination with the Birkhoff-von Neumann decomposition) to study fair allocation with an unequal number of goods and agents. The work of~\cite{AFS+23best} focuses on computing a randomized allocation with desirable ex-ante and ex-post \emph{envy-freeness} guarantees. By contrast, our work uses the technique of~\cite{AFS+23best} to achieve an \emph{equitability} guarantee.}
    
    Let $W = pW_c+q$ for some constant $p$ and $q \neq 0$. We create $p+1$ copies of every agent, say $\{a_i^1, a_i^2, \ldots a_i^{p+1}\}$. We also add $t = \left(p+1\right)n - m$ many dummy goods to the instance which are valued at zero by everyone. Note that the new instance has an equal number of agents and goods, precisely $(p+1)n$. Now each good is represented as food, and the agent copies start eating away all the available goods that they like, all at once. By the structure of the instance, exactly $W_c$ agent copies eat the same $W$ goods at the same time, and the same speed -- at the rate of one good per unit time. In particular, at the $t^{th}$ timestep, $\nicefrac{1}{\left(W-(t-1\right)W_c)}$ fraction of the remaining good is consumed by the $t^{th}$ copy of the agent. This gives us a square matrix $Y$ with $\left(p+1\right)n$ columns as goods and $\left(p+1\right)n$ rows as copies of the agents. The entry $Y[a_i^j][g]$ corresponds to the fraction of good $g$ eaten by $j^{th}$ copy of agent $i$ at timestep $j$.
    
    In particular, in the first timestep, $W$ goods are consumed by all the first copies ($a_i^1$) of $W_c$ agents (who like them) simultaneously, each of whom eats  $\nicefrac{1}{W}$ fraction of $W$ goods each. At second timestep, $\left(1-\nicefrac{W_c}{W}\right)$ fraction of these $r$ goods remain, out of which $\nicefrac{1}{\left(W-W_c\right)}$ fraction is consumed by the second copy of all the $W_c$ agents and so on. That is,  assuming $N_g$ be the set of $W_c$ agents who like $g$, we have:
    
    $$\sum_{i \in N_g} v_i(a_i^1) = \frac{W_c}{W}$$ $$\sum_{i \in N_g} v_i(a_i^2) = W_c \left(\frac{1}{W-W_c}\left(1-\frac{W_c}{W}\right) \right) = \frac{W_c}{W}$$
    $$\sum_{i \in N_g}v_i(a_i^3)= W_c \left(\frac{1}{W-2W_c}\left(1-\frac{2W_c}{W}\right)\right)=\frac{W_c}{W}$$ $$\vdots$$ $$\sum_{i \in N_g}v_i(a_i^p)=W_c \left(\frac{1}{W-(p-1)W_c} \left(1- \frac{(p-1)W_c}{W}\right)\right) = \frac{W_c}{W}$$

    For the last agent copy $a_i^{p+1}$, the fraction of each of the $W$ goods that remain is $1- \nicefrac{pW_c}{W} = \nicefrac{(W-pW_c)}{W} = \nicefrac{q}{W}$. This is divided equally among the last copy of all the $W_c$ agents, each of them getting $\nicefrac{q}{WW_c}$ fraction.
    
    Also, $a_i^{p+1}$ eats $\nicefrac{q}{WW_c}$ of $W$ goods, thereby summing to $\nicefrac{q}{W_c}$. Now $\nicefrac{q}{W_c} < 1$, and we have $t$ dummy goods remaining to be consumed. Therefore, at this timestep, $a_i^{p+1}$ for $i \in [n]$, start eating $\left(\frac{1-\nicefrac{q}{W_c}}{t}\right)$ fraction of each of the dummy goods, thereby consuming one unit of good, in aggregate. 
    
    
    We now claim that $Y$ is a doubly stochastic matrix. To this end, we first show that the fractions in every column of $Y$ adds up to $1$. For a column $c$ corresponding to an original good $g$, summing over the $p+1$ copies of $s$ agents who like $g$, we get:
    
    $$\sum_{j=1}^{p+1} \sum_{i \in N_g} v_i(a_i^j) = \sum_{j=1}^{p} \sum_{i \in N_g} v_i(a_i^j) + v_{i}(a_i^{p+1}) = p \left(\frac{W_c}{W} \right) + W_c\left(\frac{q}{WW_c}\right) = \frac{pW_c+q}{W} = 1$$

    Now for a column $c$ corresponding to a dummy good $g$, each of the agent copies eat $\left(\frac{1-\nicefrac{q}{W_c}}{t}\right)$ fraction of the dummy goods. Since there are $n$ such agents, the sum of the fractions in column $c$ is 
    $$n\left(\frac{1-\nicefrac{q}{W_c}}{t}\right) = \frac{n(W_c-q)}{W_c(p+1)n-mW_c} = \frac{n(W_c-q)}{npW_c +nW_c - nW} = \frac{W_c-q}{W_c-q} = 1$$

    Also, it is easy to see that rows in $Y$ add up to $1$. For $j \in [1, p]$, $a_i^j$ starts eating at $j^{th}$ timestep, when $\left(1-\nicefrac{(j-1)W_c}{W}\right)$ fraction of any good $g$ remains. She eats $\nicefrac{1}{\left(W-(j-1\right)W_c)}$ fraction of $W$ such goods, that adds to $W (\frac{1}{W-(j-1)W_c}) (1-\frac{(j-1)W_c}{W}) = 1$. As for the row corresponding to $a_i^{p+1}$, it adds up to $1$ by construction.
    
    This establishes the following claim.
    
    \begin{claim}
    $Y$ is a doubly stochastic matrix. 
    \end{claim}

    By \Cref{thm:birkhoff}, $Y$ can be represented as a convex combination of permutation matrices. An illustration of this is shown in \Cref{example1}. For the final allocation, one of these permutation matrices, say $P$, is selected with probability equal to the corresponding weight. A good is allocated to $a_i^j$ if and only if $P[a_i^j][g]=1$. Finally, all the goods allocated to the copies of agent $i$ are said to be allocated to agent $i$.
    
    We now claim that the resulting allocation is \EQ1{} with optimal utilitarian welfare.
    
    \begin{claim} Every integral allocation returned by the above algorithm satisfies \EQ1{}.
    \label{clm:doublynormalized}
    \end{claim}
    
    \begin{proof} (of \Cref{clm:doublynormalized}) Since the matrices in the decomposition are permutation matrices and the number of goods is equal to the number of agents, each of the agent-copies gets exactly one good. This implies that all the agents end up with an equal number of goods, precisely, $p+1$. \\
    Since the dummy goods are consumed by only the last agent copy, therefore, every agent gets at most one dummy good in the final allocation. Also, all the original goods are allocated non-wastefully, as except for the dummy goods, agents eat only the goods that they like.
    Therefore,
    whoever ends up with a dummy good have a utility of $p$ and the remaining agent have a utility of $p+1$, resulting in an $\EQ{1}$ allocation with optimal utilitarian welfare.
    \end{proof}

    Therefore the price of equity for doubly normalised instances is $1$. This finishes the proof of \Cref{thm:doubly-normalised}.
    \end{proof}


\begin{example}
\label{example1}
    Consider an instance with $4$ agents, $\{a_1, \ldots a_4\}$ and $6$ goods, $\{1, 2, \ldots 6\}$. We set $W=3$ and $W_c=2$. $a_1$ likes $\{1, 2, 3\}$, $a_2$ likes $\{4, 5, 6\}$, $a_3$ likes $\{2, 3, 4\}$ and $a_4$ likes $\{1, 5, 6\}$. Here, since $p=1$, we create $p+1 = 2$ copies of every agent, and introduce $(p+1)n-m=2\cdot4-6 = 2$ dummy goods. The corresponding matrix $Y$ and its decomposition is as follows.

\footnotesize
\begin{gather*}
\renewcommand*{\arraystretch}{0.8}
\begin{pmatrix}
\nicefrac{1}{3} &   \nicefrac{1}{3}     &    \nicefrac{1}{3}   &  & &  & 0 & 0   \\
         \nicefrac{1}{6}  &     \nicefrac{1}{6}     &    \nicefrac{1}{6}  &  & &  & \nicefrac{1}{4}  & \nicefrac{1}{4}    \\
 &          &  &  \nicefrac{1}{3} & \nicefrac{1}{3} & \nicefrac{1}{3}  & 0 & 0       \\
  &          &  &  \nicefrac{1}{6} & \nicefrac{1}{6} & \nicefrac{1}{6}  & \nicefrac{1}{4} & \nicefrac{1}{4}       \\
 & \nicefrac{1}{3} & \nicefrac{1}{3} & \nicefrac{1}{3}    &   &     & 0 & 0    \\
  & \nicefrac{1}{6} & \nicefrac{1}{6} & \nicefrac{1}{6}    &   &     & \nicefrac{1}{4} & \nicefrac{1}{4}    \\
 \nicefrac{1}{3} &          &   & & \nicefrac{1}{3} & \nicefrac{1}{3}   & 0 & 0     \\
 \nicefrac{1}{6} &          &   & & \nicefrac{1}{6} & \nicefrac{1}{6}   & \nicefrac{1}{4} & \nicefrac{1}{4}
\end{pmatrix}
=
\renewcommand*{\arraystretch}{0.1}
0.16
\begin{pmatrix} 1&&&&&&&  \\ 
&1&&&&&&  \\
&&&1&&&& \\
&&&&1&&& \\
&&1&&&&& \\
&&&&&&1& \\
&&&&&1&& \\
&&&&&&&1
\end{pmatrix}
 +
 0.083
 \begin{pmatrix} 1&&&&&&&  \\ 
&&1&&&&&  \\
&&&1&&&& \\
&&&&&1&& \\
&1&&&&&& \\
&&&&&&1& \\
&&&&1&&& \\
&&&&&&&1
\end{pmatrix} \\
\renewcommand*{\arraystretch}{0.1}
+
0.083
 \begin{pmatrix} 1&&&&&&&  \\ 
&&1&&&&&  \\
&&&1&&&& \\
&&&&&1&& \\
&1&&&&&& \\
&&&&&&&1 \\
&&&&1&&& \\
&&&&&&1&
\end{pmatrix} 
\renewcommand*{\arraystretch}{0.1}
+
0.083
 \begin{pmatrix} &1&&&&&&  \\ 
1&&&&&&&  \\
&&&&1&&& \\
&&&1&&&& \\
&&1&&&&& \\
&&&&&&&1 \\
&&&&&1&& \\
&&&&&&1&
\end{pmatrix} 
+
0.166
 \begin{pmatrix} &1&&&&&&  \\ 
&&&&&&1&  \\
&&&&1&&& \\
&&&&&&&1 \\
&&&1&&&& \\
&&1&&&&& \\
1&&&&&&& \\
&&&&&1&&
\end{pmatrix} \\
\renewcommand*{\arraystretch}{0.1}
+
0.166
 \begin{pmatrix} &&1&&&&&  \\ 
&&&&&&1&  \\
&&&&&1&& \\
&&&&&&&1 \\
&1&&&&&& \\
&&&1&&&& \\
1&&&&&&& \\
&&&&1&&&
\end{pmatrix} 
+
0.083
\begin{pmatrix} &&1&&&&&  \\ 
&&&&&&&1  \\
&&&&&1&& \\
&&&&&&1& \\
&1&&&&&& \\
&&&1&&&& \\
1&&&&&&& \\
&&&&1&&&
\end{pmatrix} 
+
0.166
\begin{pmatrix} &&1&&&&&  \\ 
&&&&&&&1  \\
&&&&&1&& \\
&&&&&&1& \\
&&&1&&&& \\
&1&&&&&& \\
&&&&1&&& \\
1&&&&&&&
\end{pmatrix} \\
\end{gather*}
\end{example}

\subsection{Graphs for comparing bounds on the PoE for binary additive valuations.}
\label{sec:figures}

To help in visualising our upper and lower bounds on the PoE (as presented in~\Cref{tab:Results}), we also present the PoE bounds graphically below (see \Cref{Fig:p=1,Fig:p=0,Fig:p=-1,Fig:p=-10}). Each graph shows the lower and upper bounds obtained as a function of $r$, the number of agent types. In each figure, the upper line in blue represents the upper bound, and the lower line in orange represents the lower bound. We provide the plots for four values of $p$, namely $p=1$ (the utilitarian welfare), $p=0$ (the Nash social welfare), $p=-1$, and $p=-10$ (recall that for $p \rightarrow -\infty$, the egalitarian welfare, the PoE is 1).

\begin{figure}[!htb]
\begin{minipage}{0.48\textwidth}
     \centering
     \includegraphics[width=.95\linewidth]{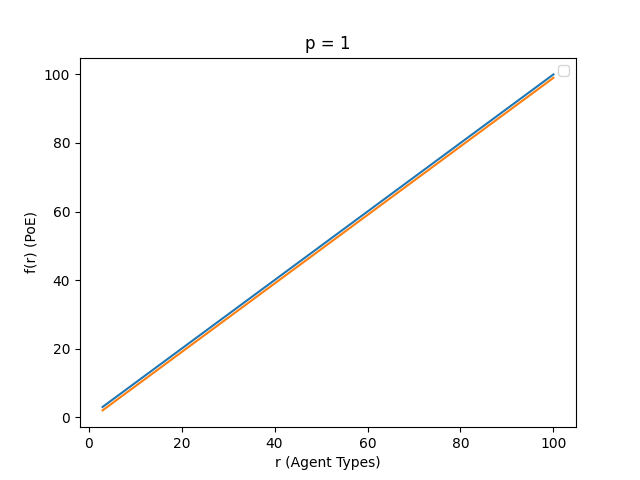}
     \caption{PoE as a function of $r$ for $p=1$}\label{Fig:p=1}
   \end{minipage}\hfill
   \begin{minipage}{0.48\textwidth}
     \centering
     \includegraphics[width=.95\linewidth]{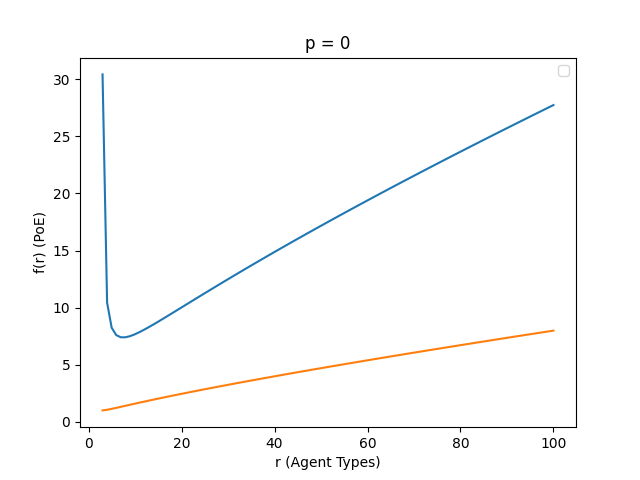}
     \caption{PoE as a function of $r$ for $p=0$}\label{Fig:p=0}
   \end{minipage}

   \begin{minipage}{0.48\textwidth}
     \centering
     \includegraphics[width=.95\linewidth]{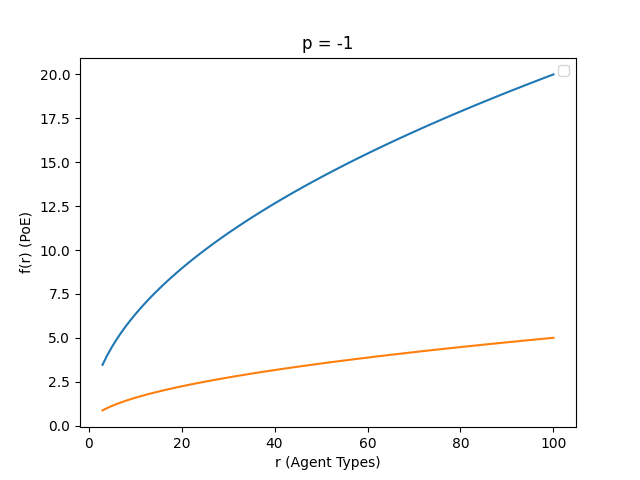}
     \caption{PoE as a function of $r$ for $p=-1$}\label{Fig:p=-1}
   \end{minipage}\hfill
   \begin{minipage}{0.48\textwidth}
     \centering
     \includegraphics[width=.95\linewidth]{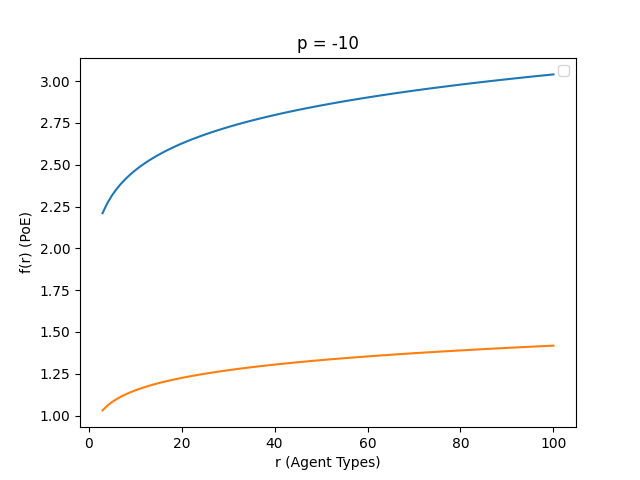}
     \caption{PoE as a function of $r$ for $p=-10$}\label{Fig:p=-10}
   \end{minipage}
\end{figure}


    
    






\end{document}